%% file: CUBIC_Markov_Approx_JNL.tex
\documentclass[10pt]{article}
\usepackage{cite}
\usepackage{graphicx}
\usepackage{siunitx}
\usepackage{mathtools}

\DeclarePairedDelimiter\floor{\lfloor}{\rfloor}
\usepackage{amsmath}
\usepackage{amsthm}
\usepackage{amssymb}
\usepackage{dsfont}
\usepackage{epsfig}
\usepackage{color}
\usepackage{url} 
\newtheorem{proposition}{Proposition}
\newtheorem{lemma}{Lemma}
\newtheorem{corollary}{Corollary}

\begin{document}
\title{An Asymptotic Approximation of TCP CUBIC}
\author{Sudheer Poojary \qquad Vinod Sharma \\ Department of ECE, Indian Institute of Science, Bangalore \\ Email: \{sudheer, vinod\}@ece.iisc.ernet.in}
\maketitle

\input{abstract.tex}
\input{introduction.tex}
\input{tcpcubic.tex}
\input{fluidmodel.tex}
\input{markovmodel.tex}
\input{simulation_results.tex}
\input{conclusion.tex}

\input{appendix.tex}

\bibliographystyle{IEEEtran} 
\bibliography{tcp-references}
\end{document}

%% file: abstract.tex
\begin{abstract}
\label{sec:abstract}
In this paper, we derive an expression for computing average window size of a single TCP CUBIC connection under random losses. Throughput expression for TCP CUBIC has been computed earlier under deterministic periodic packet losses. We validate this expression theoretically. We then use insights from the deterministic loss based model to derive an expression for computing average window size of a single TCP CUBIC connection under random losses. For this computation, we first consider the sequence of TCP CUBIC window evolution processes indexed by the drop rate, $p$ and show that with a suitable scaling this sequence converges to a limiting Markov chain as $p$ tends to $0$. The stationary distribution of the limiting Markov chain is then used to derive the average window size for small packet error rates. We validate our model and approximations via simulations.
\end{abstract}

%% file: introduction.tex
\section{Introduction}
\label{sec:introduction}
The TCP-IP protocol suite forms the backbone of the current Internet and TCP is a crucial component of it. TCP provides reliable, in-order data transfer and flow and congestion control. In this paper, we focus on TCP congestion control. TCP congestion control has been successful in preventing congestion collapse over the Internet. However in \cite{Huston2006}, \cite{rfc3649} we see that the traditional TCP congestion control algorithms can be very inefficient over wireless links and over high-speed large delay networks. A number of high-speed TCP congestion control algorithms have been proposed to address the issue of inefficiency, some notable examples being H-TCP, BIC, CUBIC, Compound and FAST \cite{Afanasyev2010}. In this paper, we consider TCP CUBIC congestion control as it is widely used. TCP CUBIC is the default congestion control algorithm on Linux since $2006$. In \cite{Yang2014}, the authors report that of the $30000$ web-servers that they considered, more than $25\%$ used TCP CUBIC. 

We first give a brief overview of the literature on traditional Additive Increase Multiplicative Decrease (AIMD), TCP which  has been extensively studied using a wide variety of tools. In \cite{Bonald1999, Liu2003}, the authors use fluid models to analyze TCP performance. In \cite{Bonald1999}, the author compares the performance of TCP Reno with TCP Vegas using a differential equation based model for TCP window evolution, whereas in \cite{Liu2003}, the authors solve for throughput of a large number of TCP Reno, New Reno and SACK flows going through AQM routers. In \cite{Vojnovic2000, Kunniyur2003}, the authors look at optimization based techniques for performance analysis of TCP. In \cite{Vojnovic2000}, the authors show that rate-distribution of TCP-like sources in a general network is given as a solution to a global optimization problem. In \cite{Kunniyur2003} the authors formulate the rate allocation problem as a congestion control game and show that the Nash equilibrium of the game is a solution to a global optimization problem. In \cite{Reddy2004}, the authors consider providing QoS to TCP and real time flows through use of rate control for the real time flows and RED at the bottleneck queues. In \cite{Mathis1997}, the authors provide expressions for TCP Reno throughput using a simple periodic-loss model. In  \cite{Padhye2000}, the authors use Markovian models to derive an expression for TCP Reno throughput under random losses.

In \cite{Xue2014}, \cite{Jain2011}, we see experimental evaluation of high speed TCP variants. The reference \cite{Xue2014} compares the performance of CUBIC, HSTCP and TCP SACK in a $10$ Gbps optical network. In \cite{Jain2011}, the authors perform an experimental evaluation of CUBIC TCP in a small buffer regime. The reference \cite{Weigle2006} is a comprehensive simulation based analysis of high speed TCP variants, where they compare the protocols for intra-protocol and inter-protocol fairness. There are many references on simulation/experimental analysis of TCP CUBIC, however there are fewer analytical results. In \cite{Sudheer2011} and \cite{Bao2010}, the authors use Markov chain based models for TCP CUBIC throughput computations. In \cite{Belhareth2013}, the authors analyze performance of TCP CUBIC in a cloud networking environment using mean-field.

In this paper, we derive throughput expression for a single TCP CUBIC flow with random losses. Throughput expressions for TCP CUBIC have been evaluated under a deterministic loss model in \cite{Ha2008}. Also average window size for TCP CUBIC  with random losses has been numerically computed using Markov chains in \cite{Sudheer2011} and \cite{Bao2010}. In \cite{Sudheer2013Allerton}, we see that the expressions for throughput in \cite{Ha2008} are not accurate when compared against the Markov chain based results in \cite{Sudheer2011}. However, the Markov chain based results do not yield a closed form expression and we need to solve for the stationary distribution of a Markov chain for each value of drop rate, $p$. For small $p$, this could be computationally expensive as the state space of the Markov chain could be very large. We address this drawback of the Markov chain model in this paper. In this paper, we get an approximation for TCP CUBIC under random losses as a function of $p$ and round trip time (RTT). We first validate the expression for TCP CUBIC throughput (given in \cite{Ha2008}) under deterministic periodic losses. We then consider the sequence of the TCP CUBIC window evolution processes indexed by the drop rate, $p$ and show that with a suitable scaling this sequence converges to a limiting Markov chain as $p$ tends to $0$. The appropriate scaling is obtained from the deterministic periodic loss model. The stationary distribution of the limiting Markov chain gives us the desired approximation. Our approach is based on a similar result used for TCP Reno throughput computation in \cite{Dumas2002}. However our proofs are significantly different.

The organization of the paper is as follows. In Section \ref{sec:tcpcubic}, we describe our system model. In Section \ref{sec:fluidmodel}, we validate the deterministic loss model expression. 
In Section \ref{sec:markovmodel}, we show that for $p > 0$ and with $W_{max} = \infty$, the window size process at RTT epochs, the window size process at loss epochs and the time between the loss epochs have unique stationary distributions and that their means under stationarity are also finite. In Section \ref{sec:asymptoticapprox}, we derive an approximation for mean window size under random losses. In Section \ref{sec:simulation_results}, we compare our model predictions against ns2 simulations. Section \ref{sec:conclusion} concludes our paper.

%% file: tcpcubic.tex
\section{System model for TCP CUBIC}
\label{sec:tcpcubic}

The window size evolution of TCP CUBIC is based on the time since last congestion epoch. The window size, (say $W_0$) at the last epoch is considered as an equilibrium window size. The TCP CUBIC window update is conservative near $W_0$ and is aggressive otherwise. The aggressive behaviour gives TCP CUBIC higher throughput compared to traditional TCP in high speed networks. The TCP CUBIC window size at time $t$, assuming $0$ to be a loss epoch and no further losses in $(0,t]$ is given by
\begin{equation}
\label{eqn:tcp_CUBIC}
W_{cubic}(t) = C\Bigl(t - \sqrt[3]{\frac{\beta W_0}{C}}\Bigr)^3 + W_0,
\end{equation}
where $W_0$ is the initial window size and $C$ and $\beta$ are TCP CUBIC parameters. The TCP CUBIC update can be slower as compared to TCP Reno. To ensure a worst case behaviour like TCP Reno, the aggregate window update is given by $W(t) = \max\{W_{cubic}(t), W_{reno}(t)\}$
where $W_{reno}(t)$ is given by
\begin{equation}
\label{eqn:tcpf}
W_{reno}(t) = W_{0}  (1 - \beta) + 3  \frac{\beta}{2 - \beta}  \frac{t}{RTT}.
\end{equation}
In our analysis henceforth, we ignore the Reno-mode operation focussing only on the CUBIC-mode. However, we account for the Reno-mode operation in the final average window size expression. We consider a single TCP CUBIC connection going through a link with constant RTT (round trip time) as shown in Figure \ref{fig:singleTCP_singleLink}. The packets of the connection may be subject to channel losses. We assume that a packet can be lost independently of other packets with probability $p$. This is a common assumption also made in \cite{Sudheer2011}, \cite{Bao2010}. Our objective is to compute an expression for TCP CUBIC throughput in this setup, which we develop in Section \ref{sec:markovmodel}. In Section \ref{sec:fluidmodel}, we discuss a deterministic loss model for TCP CUBIC and use the results developed therein  in Section \ref{sec:markovmodel} to compute TCP CUBIC average window size. 
\begin{figure}
\centering
\includegraphics[scale=0.55]{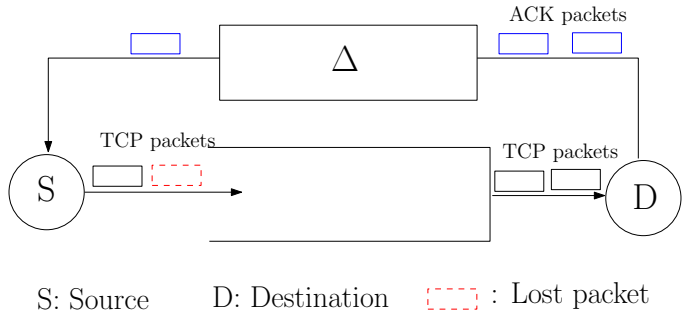}
\caption{Single TCP with fixed RTT}
\label{fig:singleTCP_singleLink}
\end{figure}

%% file: fluidmodel.tex
\section{Fluid models for TCP CUBIC}
\label{sec:fluidmodel}

We now consider a simple fluid model for TCP CUBIC. For the fluid model, we disregard the discrete nature of the TCP window size and also assume that the window update is continuous instead of happening at discrete intervals of time. The model that we consider here is a widely used deterministic-loss model (see \cite{Ha2008}, \cite{Tan2006Infocom} and \cite{Mathis1997}) used to compute the `response function' of TCP. The TCP response function is an expression for TCP throughput in terms of system parameters such as drop rates $p$ and RTT, $R$. 

We consider a single TCP CUBIC flow with constant RTT, $R$. Each packet can be dropped independently of the others with probability $p$. Under this assumption, the mean number of packets sent between losses is $\frac{1}{p}$.

We now consider the TCP window evolution under a deterministic-loss model with loss rate $\frac{1}{p}$. Let us denote the window size for the deterministic-loss model at time $t$ by $\hat{W}(t)$. Suppose $\hat{W}(0) = x$. Let $\tau_p(x)$ denote the time taken by the $\hat{W}(t)$ process to send $\frac{1}{p}$ packets with initial window size $x$, i.e., $\tau_p(x)$ satisfies
\begin{equation}
\frac{1}{R} \int_0^{\tau_p(x)} \hat{W}(t) dt = \frac{1}{p}.
\end{equation}
At $t$ $= \tau_p(x)$, $\hat{W}(t)$ undergoes a window reduction so that $\hat{W}(\tau_p(x)^+)$ $=(1 - \beta)\hat{W}(\tau_p(x))$, where $\beta$ is the multiplicative drop factor. Next, the window size $\hat{W}(t)$ evolves as given by \eqref{eqn:tcp_CUBIC} but now with initial window size, $\hat{W}(\tau_p(x))$. Again at time $t$ $=\tau_p(x) + \tau_p(\hat{W}(\tau_p(x)))$, $\hat{W}(t)$ process undergoes another loss. This process continues. 

Suppose there exists a $x^*_p$ such that $\hat{W}(\tau_p(x^*_p))$ $=x^*_p$, i.e., the fixed point equation
\begin{equation}
\label{eqn:fixed_point_form}
\hat{W}(\tau_p(x)) = x,
\end{equation}
has a unique solution. Then, if we start from $x^*_p$, the process $\hat{W}(t)$ will have a periodic behaviour with period $\tau_p(x^*_p)$ and $\hat{W}(t) \in [(1 - \beta) x^*_p, x^*_p]$.  The long time average for the process $\hat{W}(t)$ is then given by
\begin{equation}
\label{eqn:generic_rf}
 \frac{1}{\tau_p(x^*_p)} \int_{0}^{\tau_p(x^*_p)}  \hat{W}(t) dt,
\end{equation}
with $\hat{W}(0) = x^*_p$. Using the above model, the average window size for TCP CUBIC is given by  
\begin{equation}
\label{eqn:cubic}
\mathbb{E}[W(p)] = \sqrt[4]{\frac{C(4 - \beta)}{4\beta} \bigl( \frac{R}{p} \bigr)^3}.
\end{equation}
The throughput of the TCP connection is given by $\frac{\mathbb{E}[W(p)]}{R}$. 
In Proposition \ref{prop:CUBIC}, we provide a theoretical justification validating the use of the above expression for mean window size. We prove that starting from any initial window size, under the deterministic loss model with fluid window sizes, the window evolution for TCP CUBIC is eventually periodic with \eqref{eqn:cubic} giving the correct time average window size. In Proposition \ref{prop:CUBIC}, we ignore the slow start phase and ignore that there may be an upper bound on the maximum window size. These assumptions are also made by \cite{Ha2008}, \cite{Tan2006Infocom} and \cite{Mathis1997}. 

\begin{proposition}
\label{prop:CUBIC}
For the deterministic loss model, for any given $p \in (0,1)$, there exists a unique $x$ (denoted by $x^*_p$) such that $\hat{W}(\hat{\tau}_p(x)) = x$. For any $x \geq 1$ such that $\hat{W}(0) = x$, $\hat{W}(t)$ to $x^*_p$ at drop epochs. 
\end{proposition}
\begin{proof}
\subsection*{Existence of $x^*_p$}
Assuming the initial window size to be $x$, we have
\begin{equation*}
\hat{W}(\hat{\tau}_p(x)) = C\bigl(\hat{\tau}_p(x) - \sqrt[3]{\frac{\beta x}{C}}\bigr)^3 + x.
\end{equation*}
Solving for the fixed point, $x^*_p$  of $\hat{W}(\hat{\tau}_p(x))$ gives us $\hat{\tau}_p(x^*_p) = \sqrt[3]{\frac{\beta x^*_p}{C}}$. Since $\frac{1}{p}$ packets are sent in $\hat{\tau}_p(x)$, we have
\begin{equation*}
\frac{1}{R} \int_{0}^{\hat{\tau}_p(x^*_p)} \hat{W}(u) du  = \frac{1}{p}. 
\end{equation*}
The fixed point, $x^*_p$ for $\hat{W}(\hat{\tau}_p(x))$ is then given by
\begin{equation}
\label{eqn:CUBIC_xstar}
x^*_p =  \sqrt[4]{\frac{C}{\beta}} \bigl( \frac{4}{(4 - \beta}  \frac{R}{p} \bigr)^{\frac{3}{4}}.
\end{equation}
Thus for every $p \in (0,1)$, there exists a unique $x^*_p$ given by \eqref{eqn:CUBIC_xstar} such that $\hat{W}(\hat{\tau}_p(x^*_p)) = x^*_p$.
 
\subsection*{Convergence to $x^*_p$}
Let us denote the deterministic process, $\hat{W}(u)$ at time $u > 0$ with $\hat{W}(0) = x$ by $\hat{W}(u, x)$ so as to also include the initial window size in the process description explicitly. We will show convergence of the map $x \rightarrow \hat{W}(\tau_p(x))$ to the fixed point in two steps. We define $J(x) = \sqrt[3]{\frac{\beta x}{C}}$ to be the time taken by $\hat{W}(t)$ to hit $x$ given that initial window size, $\hat{W}(0)$ before drop was $x$ and there are no losses in $(0, u]$, with $u > J(x)$. 

\subsubsection*{Step 1:} We first show that if $x < x^*_p$, then $x < \hat{W}(\tau_p(x), x)$. Since $x < x^*_p$, $J(x) < J(x^*_p)$. For $t < J(x)$, we have $(t - J(x))^2 < (t - J(x^*_p))^2$ which implies $\frac{d\hat{W}(t,x)}{dt} < \frac{d\hat{W}(t,x^*_p)}{dt}$. Also $x = \hat{W}(0,x) < \hat{W}(0,x^*_p) = x^*_p$. 
Therefore for $t < J(x)$, $\hat{W}(t,x) < \hat{W}(t,x^*_p)$. Hence, we have 
\begin{equation*}
\int_{0}^{J(x)} \hat{W}(u,x) du < \int_{0}^{J(x)} \hat{W}(u,x^*_p) du < \int_{0}^{J(x^*_p)} \hat{W}(u,x^*_p) = \frac{R}{p}.
\end{equation*}
The second inequality comes due to $\hat{W}(u,x) > 0$ for all $u, x$. Therefore we get
\begin{equation*}
\int_{0}^{J(x)} \hat{W}(u,x) du < \int_{0}^{\tau_p(x)} \hat{W}(u,x) du = \frac{R}{p}
\end{equation*}
and $x = \hat{W}(J(x), x) <  \hat{W}(\tau_p(x), x)$. This shows that, if $x < x^*_p$, the window size at loss  epochs increases.

\subsubsection*{Step 2:} We now show that if $x > x^*_p$, then $ x^*_p < \hat{W}(\tau_p(x), x) < x$. The proof for $ \hat{W}(\tau_p(x), x) < x$ follows as in the previous proof and hence we do not show it here. 

Now, we prove that if $x > x^*_p$ then $x^*_p < \hat{W}(\tau_p(x), x) $. Suppose $T_1(x)$ denotes the time when $\hat{W}( T_1(x), x) = x^*_p$. From \eqref{eqn:tcp_CUBIC}, we get $T_1(x) = J(x) + \sqrt[3]{\frac{x^* - x}{C}}$. Therefore, 
\begin{equation*}
\int_{0}^{T_1(x)} \hat{W}(u,x) du = \frac{C}{4} \Bigl( (\frac{x^* - x}{C})^{\frac{4}{3}} - J(x)^4 \Bigr) + \Bigl(J(x) + \sqrt[3]{\frac{x^* - x}{C}} \Bigr)x.
\end{equation*}
Substituting $x = \alpha x^*_p$ ($\alpha > 1$) and then using \eqref{eqn:CUBIC_xstar} for $x^*_p$ simplifies the above expression to 
\begin{equation*}
\int_{0}^{T_1(x)} \hat{W}(u,x) du = \frac{R}{p} \Bigl( \frac{(1 - \alpha)^{\frac{4}{3}}}{(4 - \beta) \beta^{\frac{\beta}{3}}} + \alpha^{\frac{4}{3}} + \frac{4 \alpha (1 - \alpha)^{\frac{1}{3}}}{(4 - \beta) \beta^{\frac{\beta}{3}}} \Bigr).
\end{equation*}
Now substitute $\gamma = (\alpha - 1)$, $\gamma > 0$ and use $(4 - \beta) \beta^{\frac{1}{3}} <= 3$ for $\beta \in (0,1)$ to get
\begin{equation*}
\int_{0}^{T_1(x)} \hat{W}(u,x) du < \frac{R}{p} \Bigl(  (1 + \gamma)^{\frac{4}{3}} - \gamma^{\frac{4}{3}} - \frac{4}{3} \gamma^{\frac{1}{3}} \Bigr).
\end{equation*}
Using Lemma \eqref{lemma:mathstackex} below, we get $k \in (1,2)$, $(1+x)^k - x^k - kx^{k-1} < 1$. Therefore we have 
\begin{equation*}
\int_{0}^{T_1(x)} \hat{W}(u,x) du < \frac{R}{p} = \int_{0}^{\tau_p(x)} \hat{W}(u,x)
\end{equation*}
and $x^*_p = \hat{W}(T_1(x),x) <  \hat{W}(\tau_p(x),x)$.

Thus we have show that for any $x > x^*_p$, the window size at drop epochs (just after loss) monotonically decreases to $x^*_p$. Also, for any $x < x^*_p$, since $x < \hat{W}(\tau_p(x))$, the window size at drop epochs (just after loss) either monotonically increases to $x^*_p$ or exceeds $x^*_p$ at some iteration and eventually decreases to $x^*_p$.  Once $\hat{W}$ reaches $x^*_p$ at a drop epoch, $\hat{W}$ at all drop epochs henceforth will be $x^*_p$ and the evolution of $\hat{W}(t)$ becomes periodic.
\end{proof}

\begin{lemma}
\label{lemma:mathstackex}
If $k \in (1,2)$, we have 
\begin{equation*}
(1+x)^k - x^k < 1 + kx^{k-1},
\end{equation*}
for $x > 0$.
\end{lemma}
\begin{proof}
Consider the function $f(y) = (1+y)^k - y^k$. The second derivative of $f$, 
\begin{equation*}
\label{eqn:lemma_eqn}
f^{(2)}(y) = k(k-1)(1 + y)^{k-2} - k(k-1)y^{k-2},
\end{equation*}
is strictly less than $0$ for all $y > 0$. Therefore $f$ is a strict concave function over $(0,\infty)$.

The tangent to the curve $f(y)$ at $y = 0$ is given by $g(y) = 1 + ky$. Now since the function, $f$ is strictly concave in $(0, \infty)$, we have
\begin{equation*}
(1+y)^k - y^k < 1 + ky,
\end{equation*}
for $y > 0$.
Substituting $x = \frac{1}{y}$, we get
\begin{equation*}
(1+x)^k - 1 < x^k + kx^{k-1},
\end{equation*}
for $y > 0$.
Rearranging terms in the above inequality gives us the desired result.
\end{proof}
In Figure \ref{fig:TCP_CUBIC_FP_convergence_2}, we illustrate multiple iterations of the equation $\hat{W}(\tau_p(.))$ for input $x \in (0,100)$. We denote the $k^{th}$ iteration of $\hat{W}(\tau_p(.))$  by $W_p^k(.)$. We see that as the number of iterations, $k$ increases, $W_p^k(.)$ goes close to the fixed point irrespective of the starting point.
\begin{figure}
\centering
\includegraphics[scale=0.25, trim = 90 20 95 10, clip=true]{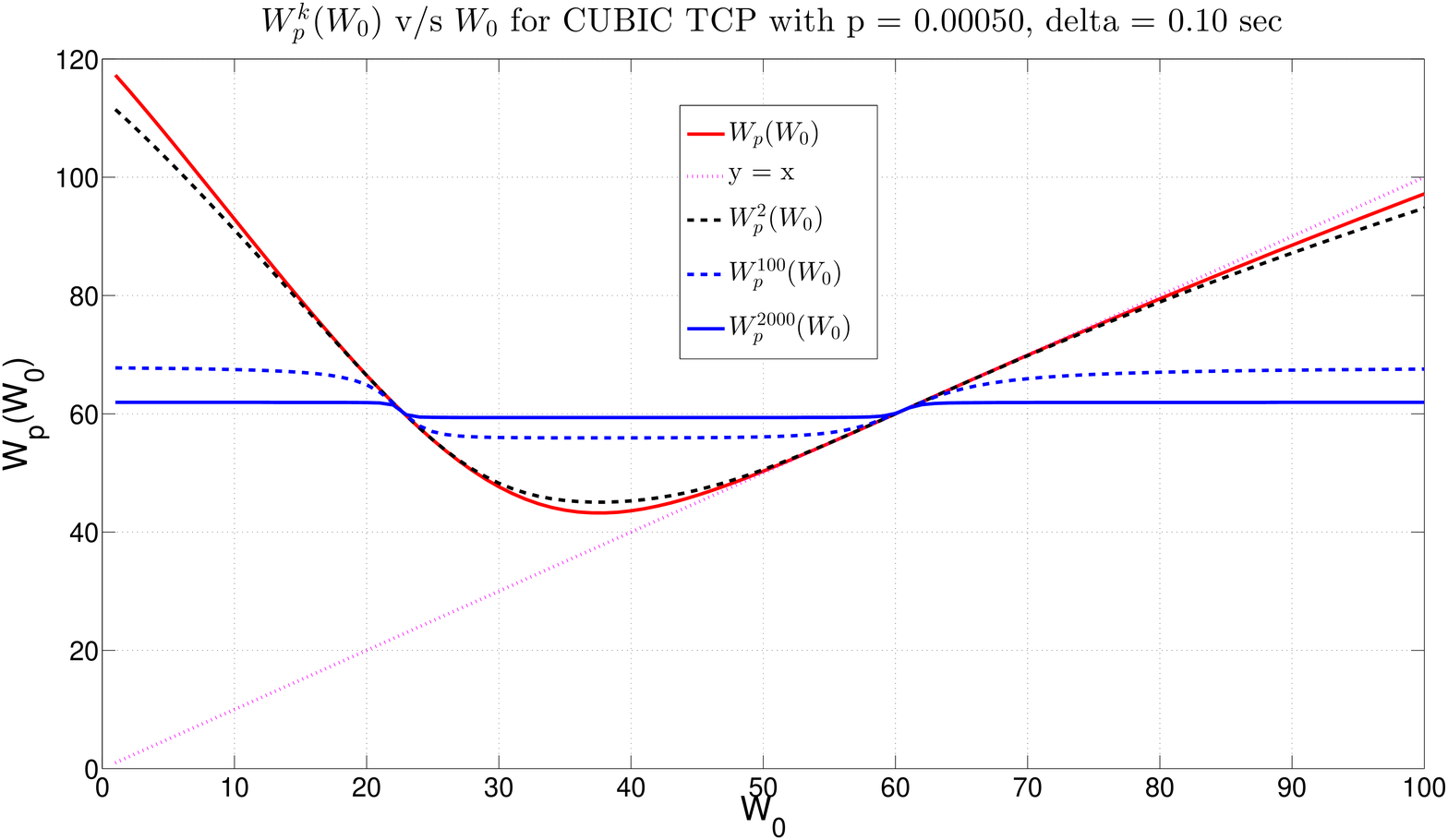}
\caption{Convergence to the fixed point of $\hat{W}(\tau_p(x)) = x$.}
\label{fig:TCP_CUBIC_FP_convergence_2}
\end{figure}

From equation \eqref{eqn:CUBIC_xstar} in Proposition \ref{prop:CUBIC}, the time between consecutive losses converges to
\begin{equation}
\label{eqn:time_between_losses_CUBIC}
\hat{\tau}_p(x^*_p) = \sqrt[3]{\frac{\beta x^*_p}{C}} = \Bigl( \frac{4 \beta R}{(4 - \beta) C p} \Bigr)^{\frac{1}{4}}.
\end{equation}
Thus for the TCP CUBIC deterministic loss model, from equation \eqref{eqn:CUBIC_xstar}, the window size at drop epochs converges to $C_1 p^{-\frac{3}{4}}$ and from equation \eqref{eqn:time_between_losses_CUBIC}, the time between consecutive losses converges to $C_2 p^{-\frac{1}{4}}$ for some constants $C_1$ and $C_2$. These are key insights which we will use in Section \ref{sec:markovmodel}, where we derive an expression for average window size under random losses.

%% file: markovmodel.tex
\section{Model with Infinite Maximum Window Size}
\label{sec:markovmodel}
We consider a single TCP connection with constant RTT, i.e., negligible queuing. We assume that the packets are dropped independently with probability $p$.  We have analyzed this system using Markov chains in \cite{Sudheer2011}. In \cite{Sudheer2011}, we derive expressions for average window size numerically when the window size $W_n$ is bounded by some $W_{max} < \infty$. We now derive an approximate expression for average window size for low packet error rates assuming $W_{max} = \infty$. 

Let $W_n(p)$ denote the window size at the end of the $n^{th}$ RTT. Let $W_n^{\prime}(p)$ denote the window size at the last drop epoch before $n$ (excluding time epoch $n$) and let $T_n(p)$ be the time elapsed between the last drop epoch before $n$ and time $n$. As in the deterministic loss model case, we ignore the Reno mode of operation and consider \eqref{eqn:tcp_CUBIC} for window evolution. The process $\{W_n^{\prime}(p),T_n(p)\}$ forms a Markov chain. We show that for $p \in (0,1)$, the processes $\{W_n(p)\}$ and $\{W_n^{\prime}(p)\}$ have unique stationary distributions.

\begin{proposition}
\label{prop:stat_dist_conv_nz_p}.
For any fixed $p\in (0,1)$ the Markov chain $\{W_n^{\prime}(p),T_n(p)\}$ has a single aperiodic, positive recurrent class with remaining states being transient. Hence it has a unique stationary distribution.
\end{proposition}
\begin{proof}
From any state in the state space of the Markov chain $\{W_n^{\prime}(p),T_n(p)\}$, a sequence of packet losses will cause the Markov chain to hit the state $(1,0)$. Therefore, the state $(1,0)$ in the state space of the Markov chain $\{W_n^{\prime}(p),T_n(p)\}$ is reachable from any state in the state space with non-zero probability. The states that can be reached by $(1,0)$ form a communicating class. The remaining states in the state space are transient, since from any of these states there is a non-zero probability of hitting $(1,0)$ and hence a non-zero probability of never returning back. 

We now show that the communicating class containing $(1,0)$ is positive recurrent. For convenience, we drop the $p$ from our notation. For a state $(z,d)$ in the communicating class, we define the Lyapunov function, $L(z,d) = z + d^4.$ The conditional one step drift of the Lyapunov function is given by
\begin{equation}
\label{eqn:CUBIC_1_drift}
\begin{aligned}
\mathbb{E}&[L(W_{n+1}^{\prime}, T_{n+1})  - L(W_n^{\prime},T_n)| (W_n^{\prime}, T_n) = (z,d)] \\
&= (z + (d+1)^4) q(z,d) + (C(Rd - \sqrt[3]{\frac{\beta z}{C}})^3 + z) (1 - q(z,d)) - z - d^4\\
&= -d^4 + (d+1)^4 q(z,d) + C(Rd-\sqrt[3]{\frac{\beta z}{C}})^3 (1 - q(z,d)),
\end{aligned}
\end{equation}
where $q(z,d) = (1-p)^{C(Rd-\sqrt[3]{\frac{\beta z}{C}})^3 + z}$ is the probability of no loss in the $n^{th}$ RTT. 

Let us denote the one-step drift in the Lyapunov function defined in \eqref{eqn:CUBIC_1_drift} by $f(z,d)$. The quantity $C(Rd - \sqrt[3]{\frac{\beta z}{C}})^3 + z \geq (1 - \beta) z$. Therefore we have $q(z,d) \leq (1 - p)^{(1 - \beta) z}$ and $-C(Rd - \sqrt[3]{\frac{\beta z}{C}})^3 \leq \beta z$. Also, $d + 1 \leq 2d$, $\forall d \in \{ 1, 2, \cdots \}$. Thus for the one-step drift, we have 
\begin{equation}
\label{eqn:CUBIC_1_drift-1}
\begin{split}
f(z,d) &= - d^4 + C(Rd - \sqrt[3]{\frac{\beta z}{C}})^3 + ((d + 1)^4  - C(Rd -\sqrt[3]{\frac{\beta z}{C}})^3) q(z,d) \\
&\leq  -d^4 + C(Rd)^3  + (16d^4 + \beta z)(1 - p)^{(1 - \beta) z}. 
\end{split}
\end{equation}
For some $\epsilon > 0$, we can choose $z^*$ such that $\beta z (1-p)^{(1-\beta)z} < \epsilon$ and $16 (1-p)^{(1-\beta)z} \leq \frac{1}{2}$,  $\forall z > z^*$. Thus we have, $\forall d > 0, z > z^*$
\begin{equation*}
\label{eqn:CUBIC_1_drift-2}
f(z,d) \leq -\frac{1}{2} d^4 + CR^3d^3  + \epsilon.
\end{equation*}
We can choose $d^*$, such that $-\frac{1}{2} d^4 + CR^3d^3 < -2 \epsilon $, for all $d > d^*$. Therefore $\forall d > d^*$ and $z > z^*$, $f(z,d) < -\epsilon$. 

Consider the first equation in \eqref{eqn:CUBIC_1_drift-1}. The term $q(z,d) = (1-p)^{C(Rd-\sqrt[3]{\frac{\beta z}{C}})^3 + z}$ falls exponentially in $z$, for any fixed $d$ and falls super-exponentially in $d$, for any fixed $z$. Hence, for any fixed $z$, as $d \rightarrow \infty$, the term $((d + 1)^4  - C(Rd -\sqrt[3]{\frac{\beta z}{C}})^3) q(z,d)$  $\rightarrow 0$. As $d^4$ is asymptotically larger than any polynomial of degree $ < 4$, for any fixed $z \leq z^*$, we can choose $t(z)$ such that $f(z,d) < -\epsilon$, $\forall d > t(z)$. Similarly, for a fixed $d$, as $z \rightarrow \infty$, the term $((d + 1)^4  - C(Rd -\sqrt[3]{\frac{\beta z}{C}})^3) q(z,d) \rightarrow 0$. For a fixed $d$, for $z$ large, $\sqrt[3]{\frac{\beta z}{C}} > Rd$. Therefore, for any fixed $d \leq d^*$, we can choose $w^{\prime}(d)$ such that $f(z,d) < -\epsilon$, $\forall z > w^{\prime}(d)$. 

Thus the one-step drift, $f(z,d) < -\epsilon$ outside of a finite set for some $\epsilon > 0$. Thus by mean drift criteria for positive recurrence, the communicating class containing $(1,0)$ is positive recurrent. Also, this class is aperiodic as the state $(1,0)$ has a non-zero  probability of hitting itself in one step (a self-loop with probability $p$).
\end{proof}

We have shown above that the Markov chain, $\{W_n^{\prime}(p)$ $, T_n(p)\}$ has a unique stationary distribution. 
Let $V_k(p)$ denote the window size at the $k^{th}$ loss epoch (just after loss) and let $G_{V_k(p)}^p$ denote the time between the $k^{th}$ and $(k+1)^{st}$ loss epoch. The following corollary is a consequence of Proposition \ref{prop:stat_dist_conv_nz_p}.

\begin{corollary}
For any $p \in (0,1)$, the processes $\{W_n(p)\}$, $\{V_k(p)\}$ and $\{G_{V_k(p)}^p\}$ have unique stationary distributions.
\end{corollary}
\begin{proof}
For the process $\{W_n^{\prime},T_n\}$, consider the inter-visit times to state $(1,0)$. These epochs are regeneration epochs for the process $\{W_n^{\prime},T_n\}$ as well as for the processes $\{W_n(p)\}$, $\{V_k(p)\}$ and $\{G_{V_k(p)}^p\}$. From Proposition \ref{prop:stat_dist_conv_nz_p}, $(1,0)$ is positive recurrent. Therefore, the mean regeneration cycle length, $\mathbb{E}[\tau_{1,0}(p)]$, (for the $\{W_n^{\prime},T_n\}$ process) is finite. Since $E[\tau_{1,0}(p)]$ is also the mean regeneration cycle length for the $\{W_n(p)\}$ process, the $\{W_n(p)\}$ process has a unique stationary distribution. The regeneration cycle length for the processes $\{V_k(p)\}$ and $\{G_{V_k(p)}^p\}$ (denoted by $\tau_V(p)$) is given by the number of loss epochs between two consecutive visits to state $(1,0)$. Since in each regeneration cycle, $\tau_V(p) \leq \tau_{1,0}(p)$, we get $\mathbb{E}[\tau_V(p)] < \infty$. Hence, the processes $\{V_k(p)\}$ and $\{G_{V_k(p)}^p\}$ also have unique stationary distributions.
\end{proof}

In Proposition \ref{prop:CUBIC_finiteEW}, we show that, for $p \in(0,1)$ TCP window size under stationarity has finite mean.
\begin{proposition}
\label{prop:CUBIC_finiteEW}
For fixed $p \in (0,1)$ the mean window size is finite, i.e., $\mathbb{E}[W(p)]$ $< \infty$ under stationarity.
\end{proposition}
\begin{proof}
Let us denote by $V_k(p)$, the window size at the $k^{th}$ congestion epoch, just after loss. For any RTT epoch, $n$, occurring between the $(k-1)^{st}$  and $k^{th}$ loss epoch, $W_{n}(p) \leq \frac{V_{k}(p)}{(1- \beta)}$. Consider the process $\{\hat{W}_n(p)\}$, with, $\hat{W}_n(p) = \frac{V_k(p)}{1 - \beta}$ for RTT epoch $n$, occurring between the $(k-1)^{st}$  and $k^{th}$ loss epoch. Then,
\begin{equation*}
\mathbb{E}[\sum_{k=1}^{\tau_{1,0}(p)} W_k(p) | W_0 = 1, T_0 = 0] \leq \mathbb{E}[\sum_{k=1}^{\tau_{1,0}(p)} \hat{W}_k(p)| W_0 = 1, T_0 = 0].
\end{equation*}
Let $\{T_k(p)\}$ be i.i.d. with distribution, $\mathbb{P}(T_k(p) = m) = (1-p)^m p$, for $m = 1, 2, \cdots $, and independent of $\{W_n(p)\}$ process. Then,
\begin{equation*}
\begin{split}
\mathbb{E}[\sum_{k=1}^{\tau_{1,0}(p)} \hat{W}_k(p)| W_0 = 1, T_0 = 0] & \leq \frac{1}{1-\beta} \mathbb{E}[\sum_{k=0}^{\tau_{V}(p)} T_k(p) V_k(p)| W_0 = 1, T_0 = 0] \\
& = \frac{1}{1-\beta} \mathbb{E}[\tau_{V}(p)] \mathbb{E}[T_1(p) V(p)] \\
& = \frac{1}{1-\beta} \mathbb{E}[\tau_{V}(p)] \mathbb{E}[T_1(p)] \mathbb{E}[V(p)].
\end{split}
\end{equation*}
The mean window size under stationarity is given by 
\begin{equation*}
\mathbb{E}[W(p)] = \frac{\mathbb{E}[\sum_{k=1}^{\tau_{1,0}(p)} W_k(p)| W_0 = 1, T_0 = 0] }{\mathbb{E}[\tau_{1,0}(p)]}.
\end{equation*}
Since the state $(1,0)$ (for  the $\{W_n^{\prime}, T_n\}$ process) is positive recurrent, $\mathbb{E}[\tau_{1,0}(p)] < \infty$. Thus to show $\mathbb{E}[W(p)] < \infty$, it is sufficient to show that  $\mathbb{E}[V(p)] < \infty$ under stationarity.

Since $\{V_n\}$ is a countable state space Markov chain, using the result for finiteness of stationary moments in \cite{Tweedie1983} we have, $\{V_n\}$ has finite mean if
\begin{equation}
\label{eqn:oneTweedie-CUBIC}
\sup_{i \in A} \mathbb{E}[V_1 | V_0 = i] < \infty, 
\end{equation}
and there is a $\delta > 0$ such that
\begin{equation}
\label{eqn:twoTweedie-CUBIC}
\mathbb{E}[V_1 | V_0 = i] \leq (1 - \delta) i, 
\end{equation}
for all $i \in A^c$, where $A$ is a finite set. 

Instead of showing that \eqref{eqn:oneTweedie-CUBIC} and \eqref{eqn:twoTweedie-CUBIC} holds for $\{V_n\}$ process, we will show that these equations hold for a process $\{\mathcal{V}_n\}$ such that $\mathcal{V}_n$ is stochastically larger than $V_n$ for each $n$, when $\mathcal{V}_0 = V_0$ . This would establish finiteness of expectation of the $\{V_n\}$ and $\{W_n\}$ stochastic processes.

\subsubsection*{Construction of $\{\mathcal{V}_n\}$ process}
Given the initial window size, $V_0$, the window size, $V_1$ at the first congestion epoch (just after loss) depends on the time at which the first congestion happens. The time between two congestion epochs $\overline{\tau}(i)$ is a random variable which depends on the initial window size. The probability mass function for $\overline{\tau}(i)$ given that the initial window size is $i$ is,

\begin{equation*}
P(\overline{\tau}(i) = m) = (1 - p)^i (1 - p)^{i_1}  \cdots (1 - (1 - p)^{i_{m-1}}),
\end{equation*}
where $m > 0$ and $i_m = (1-\beta) C\Bigl(Rm - \sqrt[3]{\frac{\beta i}{C(1 - \beta)}}\Bigr)^3 + i$ denotes the window size at the end of the $m^{th}$ RTT when the window process does not undergo any losses for $m$ RTTs.

We now define the process $\{\mathcal{V}_n\}$. Suppose $\mathcal{V}_0 = i$, then $\mathcal{V}_1$ is given by
\begin{equation*}
\mathcal{V}_1 = (1-\beta) C\Bigl(\tau(i)m - \sqrt[3]{\frac{\beta i}{C(1 - \beta)}}\Bigr)^3 + i,
\end{equation*}
where the probability distribution function of $\{\mathcal{\tau}(i)\}$ is given by
\begin{equation*}
\mathbb{P}(\tau(i) = m) = q_i^{m-1} (1 - q_i),
\end{equation*}
with $q_i = (1 -p)^i$. Thus, the probability of packet being dropped is smaller for the $\{\mathcal{V}_n\}$ process. The random variable $\mathcal{\tau}(i)$ is stochastically larger than $\overline{\tau}(i)$, i.e., $P(\mathcal{\tau}(i) > x) \geq P(\overline{\tau}(i) > x)$ for all $x \in \{0, 1, 2, 3, \cdots \}$. As the inter-congestion epoch is stochastically larger for $\{\mathcal{V}_n\}$ process, we have $\mathbb{E}[V_1 | V_0 = i] \leq \mathbb{E}[\mathcal{V}_1 | \mathcal{V}_0 = i]$. Thus it is sufficient to prove \eqref{eqn:oneTweedie-CUBIC} and \eqref{eqn:twoTweedie-CUBIC} for $\{\mathcal{V}_n\}$.

For the $\mathcal{V}_n$ process, we have
\begin{equation}
\label{eqn:CUBIC_finiteness_1}
\begin{split}
\mathbb{E}[\mathcal{V}_1 | \mathcal{V}_0 = i] & = (1 - \beta) C E[ \tau(i)^3 R^3 - 3 \tau(i)^2 R^2 K_i + 3\tau(i) R K_i^2 - K_i^3] + i,
\end{split}
\end{equation}
where $K_i = \sqrt[3]{\frac{\beta i}{C(1-\beta)}}$. The random variable $\tau(i)$ is geometric with parameter $q_i$ and its moments are given by the following equations. 
\begin{equation*}
\label{eqn:first_moment_geom}
\mathbb{E}[\tau(i)] = \frac{1}{1-q_i},
\end{equation*}

\begin{equation*}
\label{eqn:second_moment_geom}
\mathbb{E}[\tau(i)^2] = \frac{1 + q_i}{(1 - q_i)^2},
\end{equation*}
and
\begin{equation*}
\label{eqn:third_moment_geom}
\mathbb{E}[\tau(i)^3] = \frac{1}{1-q_i} + \frac{6q_i}{(1-q_i)^3}.
\end{equation*}
The above moments are substituted in \eqref{eqn:CUBIC_finiteness_1} to obtain
\begin{equation*}
\label{eqn:CUBIC_finiteness_2}
\begin{split}
\mathbb{E}[\mathcal{V}_1 | \mathcal{V}_0 = i] &= (1-\beta) \Biggl( C R^3 \biggl( \frac{1}{1-q_i} + \frac{6q_i}{(1-q_i)^3} \biggr) \\
&\hspace*{0.5cm}- 3CR^2 \Bigl(\frac{\beta i}{C(1-\beta)} \Bigr)^{\frac{1}{3}} \frac{1+q_i}{(1-q_i)^2} \\
&\hspace*{0.5cm}+ 3CR \Bigl(\frac{\beta i}{C(1-\beta)} \Bigr)^{\frac{2}{3}} \frac{1}{1-q_i} \Biggr) - \beta i + i \\
& \leq (1-\beta) \Biggl( C R^3 \biggl( \frac{1}{1-q_i} + \frac{6q_i}{(1-q_i)^3} \biggr) \\  
&\hspace*{0.5cm}+  3CR \Bigl(\frac{\beta i}{C(1-\beta)} \Bigr)^{\frac{2}{3}} \frac{1}{1-q_i} \Biggr) - \beta i + i.
\end{split}
\end{equation*}
As $i \rightarrow \infty$, $q_i \rightarrow 0$ if $p < 1$. Hence we can choose an $i^*$ such that $\forall i > i^*$, $\frac{\beta i}{2}  \geq CR^3 ( \frac{1}{1-q_i} + \frac{6q_i}{(1-q_i)^3} )$ and $\frac{\beta i}{2} \geq 3CR \Bigl(\frac{\beta i}{C(1-\beta)} \Bigr)^{\frac{2}{3}} \frac{1}{1-q_i} $. Thus we have, for all $i > i^{*}$,
\begin{equation*}
\label{eqn:CUBIC_finiteness_3}
\begin{split}
\mathbb{E}[\mathcal{V}_1 | \mathcal{V}_0 = i] & \leq (1 - \beta) \biggl(\frac{\beta i}{2}  + \frac{\beta i}{2}\biggr) - \beta i + i \\
& = (1 - \beta^2)i
\end{split}
\end{equation*}
Thus we have shown that \eqref{eqn:twoTweedie-CUBIC} holds for all $i > i^{*}$ for the $\{\mathcal{V}_n\}$ process. Also, since the random variable $\tau(i)$ has finite moments, from \eqref{eqn:CUBIC_finiteness_1} we see that \eqref{eqn:oneTweedie-CUBIC} holds for $A = \{i: i \leq i^*\}$ for the $\{\mathcal{V}_n\}$ process. This shows that under stationarity, $\mathbb{E}[\mathcal{V}(p)] < \infty$ which proves the finiteness of $\mathbb{E}[W(p)]$.
\end{proof}

\section{Asymptotic Approximations}
\label{sec:asymptoticapprox}
We now derive an expression for average window size with random losses. In Section \ref{subsec:asymptoticVk}, we derive approximations for $\{V_k(p)\}$ which will be used in Section \ref{subsec:asymptoticEW} to obtain results for $\{W_n(p)\}$  and throughput. 

\subsection{Asymptotic Approximations for $\{V_k(p)\}$}
\label{subsec:asymptoticVk}
Consider $p^{\frac{1}{4}} G_{\floor{\frac{x}{p^{\frac{3}{4}}}}}^p$, which denotes the product of $p^{\frac{1}{4}}$ with the time for first packet loss when the initial window size is $\floor{\frac{x}{p^{\frac{3}{4}}}}$\footnote{$\floor{\cdot}$ denotes the floor() operation.}. The choice of the parameters $p^{\frac{1}{4}}$ and $p^{\frac{3}{4}}$ is motivated from the deterministic-loss model in Section \ref{sec:fluidmodel} wherein the time between losses is inversely proportional to $p^{\frac{1}{4}}$ and the time average window size is inversely proportional to $p^{\frac{3}{4}}$. In Proposition \ref{prop:Gbarx_CUBIC}, we show that the term  $p^{\frac{1}{4}} G_{\floor[\big]{\frac{x}{p^{\frac{3}{4}}}}}^p$ converges to a random variable $\overline{G}_x$ as $p \rightarrow 0$ for all $x \geq 1$.

\begin{proposition}
\label{prop:Gbarx_CUBIC}
For $x > 0$, as $p \rightarrow 0$, $p^{\frac{1}{4}} G_{\floor{\frac{x}{p^{\frac{3}{4}}}}}^p$ converges in distribution to a random variable $\overline{G}_x$, with
\begin{equation}
\label{eqn:Gbarx_CUBIC2}
\begin{split}
\mathbb{P}(\overline{G}_x \geq y) = \exp \bigl(- xy -  \frac{CR^3y^4}{4} + \sqrt[3]{\frac{\beta x C^2}{(1 - \beta)}} y^3 R^2  - \bigl(\frac{\beta x C^{1/2}}{1 - \beta}\bigr)^{\frac{2}{3}} \frac{3Ry^2}{2} \bigr).
\end{split}
\end{equation}
Also for any finite $M$, if $x, y \leq M$ the above convergence is uniform in $x$ and $y$, i.e.,
\begin{equation}
\label{eqn:Gbarx_CUBIC3}
\lim_{p \rightarrow 0} \sup_{p^{\frac{3}{4}} \leq x \leq M, y \leq M} \Bigl|\mathbb{P}\Bigl(p^{\frac{1}{4}}G_{\floor{\frac{x}{p^{\frac{3}{4}}}}}^p \geq y\Bigr) - \mathbb{P}(\overline{G}_x \geq y) \Bigr| = 0.
\end{equation}
\end{proposition}

\begin{proof}
We have
\begin{equation*}
\mathbb{P}\Bigl(G_{\floor*{\frac{x}{p^{\frac{3}{4}}}}}^p > \floor{ \frac{y}{p^{\frac{1}{4}}}}\Bigr) \leq
\mathbb{P}\Bigl(p^{\frac{1}{4}}G_{\floor{\frac{x}{p^{\frac{3}{4}}}}}^p \geq y\Bigr) \leq \mathbb{P}\Bigl(G_{\floor{\frac{x}{p^{\frac{3}{4}}}}}^p \geq \floor{ \frac{y}{p^{\frac{1}{4}}}}\Bigr), 
\end{equation*}
with $\mathbb{P}\bigl(G_{\floor{\frac{x}{p^{\frac{3}{4}}}}}^p > \floor{ \frac{y}{p^{\frac{1}{4}}}}\bigr) = (1-p)^{x_0 + x_1 + \cdots + x_{\floor{\frac{y}{p^{\frac{3}{4}}}}}}$, where $x_0 = \floor{\frac{x}{p^{\frac{3}{4}}}}$ is the initial window size (immediately after a loss) and $x_i$ is the window size at the end of the $i^{th}$ RTT. Using \eqref{eqn:tcp_CUBIC}, we get $x_i =  \floor{C(iR - K)^3 + \frac{x_0}{(1 -\beta)}}$ with $K = \sqrt[3]{\frac{\beta x_0}{(1-\beta)C}}$. Let $m = \floor{\frac{y}{p^{\frac{1}{4}}}}$. The term $C(iR - K)^3 + \frac{x_0}{(1 -\beta)} \in [x_i, x_{i}+1]$. Therefore,
\begin{equation}
\label{eqn:Gbarx_CUBIC4}
\begin{split}
\mathbb{P}\bigl( & G_{\floor{\frac{x}{p^{\frac{3}{4}}}}}^p > \floor{ \frac{y}{p^{\frac{1}{4}}}}\bigr) 
 \in \bigl[  (1-p)^{ x_0 + C(R - K)^3 + \frac{x_0}{1 - \beta} + \cdots + C(mR -K)^3 + \frac{x_0}{1 - \beta} + (m+1)},\\
&  (1-p)^{ x_0 + (C(R - K)^3 + \frac{x_0}{1 - \beta}) + \cdots + (C(mR -K)^3 + \frac{x_0}{1 - \beta})   - (m+1)} \bigr].
\end{split}
\end{equation}
The terms on the RHS (right hand side) of equation \eqref{eqn:Gbarx_CUBIC4} can be simplified as follows,
\begin{equation*}
\begin{split}
(R - & K)^3  + (2R - K)^3 +  \cdots + (mR - K)^3 \\
= & R^3 \sum_{i=1}^{m} i^3 - 3 R^2 K \sum_{i=1}^{m} i^2 + 3RK^2 \sum_{i=1}^{m} i - mK^3.
\end{split}
\end{equation*}
Therefore,
\begin{equation}
\label{eqn:Gbarx_CUBIC5}
\begin{split}
x_0 &   + (C(R - K)^3 + \frac{x_0}{1 - \beta}) + \cdots + (C(mR -K)^3 + \frac{x_0}{1 - \beta})    \\
& =    x_0 + \frac{mx_0}{1- \beta}  + C( R^3 \sum_{i=1}^{m} i^3 - 3 R^2 K \sum_{i=1}^{m} i^2 + 3RK^2 \sum_{i=1}^{m} i - mK^3 ).
\end{split}
\end{equation}
After we expand the series, $\sum_{i = 1}^{m} i$, $\sum_{i = 1}^{m} i^2$ and  $\sum_{i = 1}^{m} i^3$, we see that the RHS of equation \eqref{eqn:Gbarx_CUBIC5} has terms of the form $K^n m^j $ with $n+j \leq 4$. Now,
\begin{equation}
\label{eqn:Gbarx_CUBIC6}
\begin{split}
\lim_{p \rightarrow 0} K^n m^j = & \lim_{p \rightarrow 0}  \Biggl( \sqrt[3]{\frac{\beta x_0}{(1 - \beta)C}} \Biggr)^n \Bigl(\frac{y}{p^{\frac{1}{4}}} \Bigr)^j \\
= & \lim_{p \rightarrow 0} \Bigl( \frac{\beta x }{ (1 - \beta) C} \Bigr)^{\frac{n}{3}} p^{ - \frac{n + j}{4}} y^j, 
\end{split}
\end{equation}
and 
\begin{equation*}
\lim_{p \rightarrow 0}  (x_0 + \frac{mx_0}{1- \beta}) = \lim_{p \rightarrow 0}  \frac{x}{p^{\frac{3}{4}}} + \frac{xy}{p(1 - \beta)}. 
\end{equation*}
Therefore, for $n+ j < 4$, $\lim_{p \rightarrow 0} (1-p)^{K^n m^j} = 1$. Also for the term $(1-p)^{m + 1}$ in equation \eqref{eqn:Gbarx_CUBIC4}, we note that $ \lim_{p \rightarrow 0}  (1-p)^{m + 1} = 1$. Therefore for the limit of equation \eqref{eqn:Gbarx_CUBIC4} as $p \rightarrow 0$, we need to only consider terms of the form $K^n m^j $ with $n+j = 4$ and the term $\frac{xy}{1 - \beta}$. For these terms, we use $\lim_{p \rightarrow 0}(1-p)^{\frac{1}{p}} = \exp(-1)$ to get,
\begin{equation}
\label{eqn:Gbarx_CUBIC7}
\mathbb{P}\bigl(G_{\floor{\frac{x}{p^{\frac{3}{4}}}}}^p > \floor{ \frac{y}{p^{\frac{1}{4}}}}\bigr) \rightarrow \exp \bigl(-  xy -  \frac{CR^3y^4}{4} + \sqrt[3]{\frac{\beta x C^2}{(1 - \beta)}} y^3 R^2  - \bigl(\frac{\beta x C^{1/2}}{1 - \beta}\bigr)^{\frac{2}{3}} \frac{3Ry^2}{2} \bigr).
\end{equation}
Using similar steps as above, we can show that
\begin{equation*}
\label{eqn:Gbarx_CUBIC7A}
\mathbb{P}\bigl(G_{\floor{\frac{x}{p^{\frac{3}{4}}}}}^p \geq \floor{ \frac{y}{p^{\frac{1}{4}}}}\bigr) \rightarrow \exp \bigl(-  xy -  \frac{CR^3y^4}{4} + \sqrt[3]{\frac{\beta x C^2}{(1 - \beta)}} y^3 R^2  - \bigl(\frac{\beta x C^{1/2}}{1 - \beta}\bigr)^{\frac{2}{3}} \frac{3Ry^2}{2} \bigr).
\end{equation*}
This proves convergence of $p^{\frac{1}{4}} G_{\floor{\frac{x}{p^{\frac{3}{4}}}}}^p$ in distribution to $\overline{G}_x$.

We now show uniform convergence of $ \mathbb{P}\Bigl(  p^{\frac{1}{4}}  G_{\floor{\frac{x}{p^{\frac{3}{4}}}}}^p \geq  y\Bigr) $ to $\mathbb{P}(\overline{G}_x \geq y) $. We assume that $x,y$ are bounded by $M$. Taking logarithms on both sides of \eqref{eqn:Gbarx_CUBIC4}, we get
\begin{equation}
\label{eqn:Gbarx_CUBIC8}
\begin{split}
\log \mathbb{P}\Bigl(& G_{\floor{\frac{x}{p^{\frac{3}{4}}}}}^p > \floor{ \frac{y}{p^{\frac{1}{4}}}}\Bigr)  \\
 \in & \Bigl[  \Bigl( x_0 + (C(R - K)^3 + \frac{x_0}{1 - \beta}) + \cdots + (C(mR -K)^3 + \frac{x_0}{1 - \beta})\\ 
& \hspace{1mm} + (m+1) \Bigr) \log (1-p),  \Bigl(x_0 + (C(R - K)^3 + \frac{x_0}{1 - \beta}) + \cdots + \\ 
& \hspace{9mm} (C(mR -K)^3 + \frac{x_0}{1 - \beta}) - (m+1)   \Bigr) \log (1-p) \Bigr].
\end{split}
\end{equation}
The equation \eqref{eqn:Gbarx_CUBIC8} has elements of the form $K^n m^j \log (1-p) $ with $n+j \leq 4$, the term $(x_0 + \frac{x_0m}{1 - \beta}) \log (1-p)$ and  $(m+1)\log (1-p)$. The elements with $n+j = 4$ are the only terms that contribute to the limit, i.e., \eqref{eqn:Gbarx_CUBIC7}.  From \eqref{eqn:Gbarx_CUBIC6}, the remaining elements in \eqref{eqn:Gbarx_CUBIC8} are of the form $c(n,j) x^{\frac{n}{3}} y^{j} p^{ - \frac{n+j}{4}} \log (1-p)$ with $n+j < 4$ and $ 0 \leq n \leq 4$ with $c(n,j)$ being some finite coefficient. If $ n+j < 4$, $p^{- \frac{n+j}{4}}$ is of the form $p^{\epsilon(n,j)}$ with $\epsilon(n,j) > -1$. These terms can be grouped together as $f(x,y,p)$, where $f$ has elements of the form $x^{\frac{n}{3}} y^{j} p^{- \frac{n+j}{4}} \log (1-p)$ with $n+j < 4$ and $ 0 \leq n \leq 4$, the element $p^{\frac{-3}{4}}x \log (1-p)$ and the element $(1 + p^{\frac{-3}{4}}y) \log (1-p)$.  Let $T = \max \{1, M\}$, hence $x^{\frac{n}{3}} y^j \leq T^{ \frac{n}{3} + j}$, for $x,y \leq M$. Therefore we have,
\begin{equation*}
\label{eqn:Gbarx_CUBIC9}
\begin{split}
\Bigl| \log \Bigl( \mathbb{P}\Bigl(  G_{\floor{\frac{x}{p^{\frac{3}{4}}}}}^p > \floor{ \frac{y}{p^{\frac{1}{4}}}}   \Bigr) - \log \mathbb{P}(\overline{G}_x \geq y) \Bigr|  & \leq |f(x,y, p)| \\ 
& \leq c_ 1 T^{4} p^\epsilon \log (1-p), \\
\end{split}
\end{equation*}
where the term $T^{4}$ in the inequality comes from the element with the largest power for $x$ and $y$ in the RHS of \eqref{eqn:Gbarx_CUBIC8} which is of the form $c y^4 x^0 $, $\epsilon$ $= \displaystyle{\min_{n,j: n + j < 4}}$ $\epsilon(n,j,k)$ $> -1$ and $c_1$ is a constant independent of $p, x, y$. Therefore, we have
\begin{equation*}
\label{eqn:Gbarx_CUBIC10}
\lim_{p \rightarrow 0} \sup_{p^{\frac{3}{4}} \leq x \leq M, y \leq M} \Bigl| \log \mathbb{P}\bigl(  G_{\floor{\frac{x}{p^{\frac{3}{4}}}}}^p > \floor{ \frac{y}{p^{\frac{1}{4}}}}  \bigr) - \log \mathbb{P}(\overline{G}_x \geq y) \Bigr|  = 0.
\end{equation*}
We can similarly prove
\begin{equation*}
\label{eqn:Gbarx_CUBIC10A}
\lim_{p \rightarrow 0} \sup_{p^{\frac{3}{4}} \leq x \leq M, y \leq M} \Bigl| \log \mathbb{P}\bigl(  G_{\floor{\frac{x}{p^{\frac{3}{4}}}}}^p \geq \floor{ \frac{y}{p^{\frac{1}{4}}}}  \bigr) - \log \mathbb{P}(\overline{G}_x \geq y) \Bigr|  = 0.
\end{equation*}
The result \eqref{eqn:Gbarx_CUBIC3} in Proposition \ref{prop:Gbarx_CUBIC} follows from the uniform continuity of the $\exp()$ function on $(-\infty, 0).$
\end{proof}

We now derive a limiting result for the $\{V_k(p)\}$ process embedded at the loss epochs of the TCP CUBIC window evolution process. The process $\{V_k(p)\}$ is a Markov chain embedded within the window size process $\{W_n(p)\}$. Let $K(x) = \sqrt[3]{\frac{\beta x}{(1 - \beta) C}}$. If $V_0(p) = \floor{\frac{x}{p^{\frac{3}{4}}}}$, then $V_1(p)$ is 
\begin{equation}
\label{eqn:CUBICVn_of_p}
V_1(p) = (1 - \beta) \Bigl( C(G_{\floor{\frac{x}{p^{\frac{3}{4}}}}}^p R - K(\floor{\frac{x}{p^{\frac{3}{4}}}}))^3 + \frac{1}{1-\beta} \floor{\frac{x}{p^{\frac{3}{4}}}} \Bigr),
\end{equation}
where $G_x^p$ denotes time (in multiples of $R$) between consecutive losses where the window size immediately after the first of these losses was $x$.

We now define a Markov chain which serves as the limit for the process $\{V_n(p)\}$ with appropriate scaling. Define a Markov chain $\{\overline{V}_n\}$ as follows. Let $\overline{V}_0$ be a random variable with an arbitrary initial distribution on $\mathbb{R}^+$. Define $\overline{V}_n$  for $n \geq 1$ as
\begin{equation}
\label{eqn:CUBIC_lim_Vn_def}
\overline{V}_n = (1 - \beta) \Bigl( C(\overline{G}_{\overline{V}_{n-1}} R - K(\overline{V}_{n-1}))^3 + \frac{\overline{V}_{n-1}}{1-\beta}\Bigr),
\end{equation}
where $\{\overline{G}_{\overline{V}_{n-1}}\}$ are random variables with distribution given by \eqref{eqn:Gbarx_CUBIC2} chosen independently of $\{\overline{V}_k: k < n-1\}$. The following proposition shows that the process $\{ \overline{V}_n \}$ defined by \eqref{eqn:CUBIC_lim_Vn_def} are the appropriate limiting distribution for the $\{ V_n(p) \}$ process as $p \rightarrow 0$. 
\begin{proposition}
\label{prop:Vbarx_CUBIC}
Suppose $\overline{V}_0 = x$ and $V_0(p) = \floor{\frac{x}{p^{\frac{3}{4}}}}$ for some $x > 0$ for all $p>0$. Then we have 
\begin{equation}
\label{eqn:Vbarx_CUBIC2}
\begin{split}
\lim_{p \rightarrow 0} \sup_{x > p^{\frac{3}{4}}}  \Bigl| & \mathbb{P}_x(p^{\frac{3}{4}}V_1(p) \leq a_1, \hspace*{1mm} p^{\frac{3}{4}}V_2(p) \leq a_2, \cdots,  \hspace*{1mm} p^{\frac{3}{4}}V_n(p) \leq a_n  ) \\
& - \mathbb{P}_x(\overline{V}_1 \leq a_1, \overline{V}_2 \leq a_2 \cdots \overline{V}_n \leq a_n)  \Bigr| = 0,
\end{split}
\end{equation}
where $ a_i \in \mathbb{R}^+$, for $i = 1, 2, \cdots, n$ and $\mathbb{P}_x$ denotes the law of the processes when $\overline{V}_0 = x$ and $V_0(p) = \floor{\frac{x}{p^{\frac{3}{4}}}}$.
\end{proposition}
\begin{proof}
We prove \eqref{eqn:Vbarx_CUBIC2} for $n=1,2$, the proof for $n > 2$ follows by induction. 

For $n=1$, 
\begin{equation*}
\label{eqn:Vbarx_CUBIC3}
\begin{split}
\lim_{p \rightarrow 0} \mathbb{P}_x(p^{\frac{3}{4}}V_1(p) \leq a_1) &= \lim_{p \rightarrow 0} \mathbb{P} \bigl(  p^{\frac{3}{4}} (1 - \beta) C\bigl(R G^p_{\floor{\frac{x}{p^{\frac{3}{4}}}}} - K\bigl(\floor{\frac{x}{p^{\frac{3}{4}}}}\bigr)\bigr)^3 + x \leq a_1 \bigr) \\
&= \lim_{p \rightarrow 0} \mathbb{P}\bigl( p^{\frac{1}{4}} G^p_{\floor{\frac{x}{p^{\frac{3}{4}}}}} \leq \frac{\bigl(\frac{a_1 - x}{C(1-\beta)}\bigr)^{\frac{1}{3}} + K(x)}{R}\bigr) \\
&= \mathbb{P}\bigl(\overline{G}_{x} \leq \frac{\bigl(\frac{a_1 - x}{C(1-\beta)}\bigr)^{\frac{1}{3}} + K(x)}{R}\bigr) \\
&= \mathbb{P}_x(\overline{V}_1 \leq a_1).
\end{split}
\end{equation*}
From equation \eqref{eqn:Gbarx_CUBIC3} in Proposition \ref{prop:Gbarx_CUBIC}, the convergence is uniform in $x$ over any bounded interval. Also, from \eqref{eqn:CUBICVn_of_p} and \eqref{eqn:CUBIC_lim_Vn_def},  for $x > \frac{a_1}{1-\beta}$, we have $\mathbb{P}_x( p^{\frac{3}{4}}V_1(p) \leq a_1) $ $=$ $\mathbb{P}_x( \overline{V}_1 \leq a_1) $ $= 0$. Therefore, 
\begin{equation*}
\label{eqn:Vbarx_CUBIC4}
\lim_{p \rightarrow 0} \sup_{x \geq p^{\frac{3}{4}} } \Bigl| \mathbb{P}_x \Bigl( p^{\frac{3}{4}}V_1(p) \leq a_1   \Bigr) - \mathbb{P}_x \Bigl( \overline{V}_1 \leq a_1  \Bigr) \Bigr| = 0.
\end{equation*}
This proves \eqref{eqn:Vbarx_CUBIC2} for $n = 1$. 

We now prove the result for $n=2$. Consider
\begin{equation*}
\label{eqn:Vbarx_CUBIC5}
\begin{split}
\mathbb{P}_x(& p^{\frac{3}{4}}V_1(p) \leq a_1, p^{\frac{3}{4}}V_2(p) \leq a_2) \\
= \displaystyle\int\limits_{0}^{a_1} & \mathbb{P} \bigl( C(1-\beta)\bigl(R p^{\frac{1}{4}} G^p_{\floor{\frac{y}{p^{\frac{3}{4}}}}} - K(y)\bigr)^3 + y \leq a_2  \bigr) \mathbb{P}_x(p^{\frac{3}{4}}V_1(p)  \in dy).
\end{split}
\end{equation*}
From equation \eqref{eqn:Gbarx_CUBIC3} in Proposition \ref{prop:Gbarx_CUBIC}, the term $\mathbb{P} (C(1-\beta)(R p^{\frac{1}{4}} G^p_{\floor{\frac{y}{p^{\frac{3}{4}}}}} - K(y))^3 + y \leq a_2)$ converges to $\mathbb{P}(C(1-\beta)(R p^{\frac{1}{4}} \overline{G}_{x} - K(y))^3 + y \leq a_2)$ uniformly in $y$ for $y < a_1$. Therefore, for any given $\epsilon >0$ there exists a $p^*$ such that for  $p < p^*$, 
\begin{equation*}
\label{eqn:Vbarx_CUBIC6}
\begin{split}
\Bigl| \mathbb{P}_x(& p^{\frac{3}{4}}V_1(p) \leq a_1, p^{\frac{3}{4}}V_2(p) \leq a_2) \\
& - \displaystyle\int\limits_{0}^{a_1}  \mathbb{P} \bigl( C(1 - \beta)\bigl(R p^{\frac{1}{4}} \overline{G}_{y} - K(y)\bigr)^3 + y \leq a_2  \bigr) \mathbb{P}_x(p^{\frac{3}{4}}V_1(p)  \in dy) \Bigr| \leq \epsilon. 
\end{split}
\end{equation*}
Now, 
\begin{equation*}
\label{eqn:Vbarx_CUBIC6A}
\begin{split}
\displaystyle\int\limits_{0}^{a_1}  \mathbb{P} & \bigl( C(1 - \beta)\bigl(R p^{\frac{1}{4}} \overline{G}_{y} - K(y)\bigr)^3 + y \leq a_2  \bigr) \mathbb{P}_x(p^{\frac{3}{4}}V_1(p)  \in dy) \\
& =  \displaystyle\int\limits_{0}^{\infty}  \mathbb{P} \bigl( C(1-\beta)\bigl(R p^{\frac{1}{4}} \overline{G}_{y} - K(y)\bigr)^3 + y \leq a_2  \bigr) \mathds{1}_{\{y \leq a_1\}} \mathbb{P}_x(p^{\frac{3}{4}}V_1(p)  \in dy) \\
& = \mathbb{E}_x [ g(p^{\frac{3}{4}} V_1(p))],
\end{split}
\end{equation*}
where the function $g(y) = \mathbb{P} \bigl( C(1-\beta)\bigl(R p^{\frac{1}{4}} \overline{G}_{y} - K(y)\bigr)^3 + y \leq a_2  \bigr) \mathds{1}_{\{y \leq a_1\}}$. For any continuous functions $f$ on $\mathbb{R}^+$ with compact support, using Proposition \ref{prop:unif_conv} from Appendix \ref{app:appendixA}, we have 
\begin{equation}
\label{eqn:Vbarx_CUBIC7}
\lim_{p \rightarrow 0} \sup_{x \geq p^{\frac{3}{4}} } \Bigl| E_x \bigl[ f(p^{\frac{3}{4}}V_1(p)) \bigr] - E_x \bigl[ f(\overline{V}_1) \bigr] \Bigr| = 0,
\end{equation}
The function $g$ is continuous with compact support. Therefore using \eqref{eqn:Vbarx_CUBIC7} we get,
\begin{equation*}
\label{eqn:Vbarx_CUBIC8}
\lim_{p \rightarrow 0} \sup_{x \geq p^{\frac{3}{4}} } \Bigl| \mathbb{P}_x \bigl( p^{\frac{3}{4}}V_1(p) \leq a_1, p^{\frac{3}{4}}V_2(p) \leq a_2   \bigr) - \mathbb{P}_x \bigl( \overline{V}_1 \leq a_1, \overline{V}_2 \leq a_2   \bigr) \Bigr| = 0.
\end{equation*}

The proof of \eqref{eqn:Vbarx_CUBIC2} for $n>2$ can be done using induction as follows. 
\begin{equation*} 
\label{eqn:Vbarx_CUBIC8A}
\begin{split}
\mathbb{P}_x(& p^{\frac{3}{4}}V_1(p) \leq a_1, p^{\frac{3}{4}}V_2(p) \leq a_2, \cdots, p^{\frac{3}{4}}V_{n+1}(p) \leq a_{n+1}) \\
= &  \displaystyle\int\limits_{0}^{a_1}  \mathbb{P}_y( p^{\frac{3}{4}}V_1(p) \leq a_2, \cdots, p^{\frac{3}{4}}V_{n}(p) \leq a_{n+1}) \mathbb{P}_x(p^{\frac{1}{2 -k}}V_1(p)  \in dy). 
\end{split}
\end{equation*}
Assuming the result holds for $n$, we have
\begin{equation*}
\label{eqn:Vbarx_CUBIC8B}
\begin{split}
\Big| \mathbb{P}_x(  p^{\frac{3}{4}}V_1(p) & \leq a_1, p^{\frac{3}{4}}V_2(p) \leq a_2, \cdots, p^{\frac{3}{4}}V_{n+1}(p) \leq a_{n+1})  
-  \mathbb{E}_x [ g_n(p^{\frac{3}{4}} V_1(p))] \Big| \\ 
\leq \epsilon,
\end{split}
\end{equation*}
where the function $g_n(y) = \mathbb{P}_y( \overline{V}_1  \leq a_2, \cdots, \overline{V}_n  \leq a_{n+1}) \mathds{1}_{\{y \leq a_1\}}$. The function $g_n(\cdot)$ is continuous by the induction hypothesis. Using Proposition \ref{prop:unif_conv} from Appendix \ref{app:appendixA}, gives us the desired result.
\end{proof}

Since the finite dimensional distributions of $ \{ p^{\frac{3}{4}}V_n(p) \} $ converge to $ \{ \overline{V}_n \}$, we have
\begin{corollary}
\label{coro:Vbarx_CUBIC}
If $\lim_{p \rightarrow 0}$  $p^{\frac{3}{4}}V_0(p)$ converges in distribution to $\overline{V}_0$, then the Markov chain $ \{ p^{\frac{3}{4}}V_n(p) \} $ converges in distribution to the Markov chain $ \{ \overline{V}_n \}$.
\end{corollary}
\begin{proof}
Let $\hat{\pi}_p$ and $\hat{\pi}$ be the initial distributions for the processes $ \{ p^{\frac{3}{4}}V_n(p) \} $ and $ \{ \overline{V}_n \}$  respectively. Also, let $\hat{\pi}_p$ converge weakly to $\hat{\pi}$. We have
\begin{equation*}
\begin{split}
\Bigl| \mathbb{P}_{\hat{\pi}_p} & (p^{\frac{3}{4}}V_1(p) \leq a_1) - \mathbb{P}_{\hat{\pi}}(\overline{V}_1  \leq a_1) \Bigr| \\
= & \Bigl|  \int_x  \mathbb{P}_x(p^{\frac{3}{4}}V_1(p) \leq a_1) \hat{\pi}_p(dx) -  \mathbb{P}_x(\overline{V}_1 \leq a_1) \hat{\pi}(dx) \Bigl| \\
\leq & \int_x  \Bigl| \mathbb{P}_x(p^{\frac{3}{4}}V_1(p) \leq a_1) - \mathbb{P}_x(\overline{V}_1 \leq a_1) \Bigr| \hat{\pi}_p(dx)\\
&+ \Bigl| \int_x  \mathbb{P}_x(\overline{V}_1 \leq a_1) \hat{\pi}_p(dx) - \int_x  \mathbb{P}_x(\overline{V}_1 \leq a_1) \hat{\pi}(dx) \Bigr|.
\end{split}
\end{equation*}
From Proposition \ref{prop:Vbarx_CUBIC}, for any $\epsilon > 0$, we can choose a $p^*$, such that for all $p < p^*$, $\Bigl| \mathbb{P}_x(p^{\frac{3}{4}}V_1(p) \leq a_1) - \mathbb{P}_x(\overline{V}_1 \leq a_1) \Bigr|$ $\leq \epsilon$, for all $x \geq p^{\frac{3}{4}}$. Also, since $ \mathbb{P}_x(\overline{V}_1 \leq a_1)$ is a continuous, bounded function in $x$, we have
\begin{equation*}
\lim_{p \rightarrow 0} \Bigl| \int_x  \mathbb{P}_x(\overline{V}_1 \leq a_1) \hat{\pi}_p(dx) -  \int_x  \mathbb{P}_x(\overline{V}_1 \leq a_1) \hat{\pi}(dx) \Bigr| = 0.
\end{equation*}
This proves the result in Corollary \ref{coro:Vbarx_CUBIC} for $n = 1$. The proof for $n \geq 2$ follows easily from induction.
\end{proof}

We now prove that the limiting Markov chain $\{ \overline{V}_n \}$ has a unique invariant distribution. For proving that, the given proof requires that (some of) the moments of $\overline{G}_x$ be uniformly bounded in $x$, which follows from the following lemma. 
\begin{lemma}
\label{lemma:MGF_Gbarx_bdd}
There exists $\zeta > 0$ such that for all $t \in (-\zeta, \zeta)$, we have
\begin{equation*}
\label{eqn:MGF_Gbarx_bdd}
\sup_{x>0} \mathbb{E}[e^{t\overline{G}_x}] < \infty.
\end{equation*}
\end{lemma}
\begin{proof}
Consider $H(y) = \sup_{x>0} \mathbb{P}(\overline{G}_x \geq y)$. The function $H(y)$ upper bounds $\mathbb{P}(\overline{G}_x \geq y)$ for all $x > 0$. From \eqref{eqn:Gbarx_CUBIC2}, 
\begin{equation}
\label{eqn:MGF_Gbarx_bdd2}
\begin{split}
H(y) = \exp \bigl(-  \inf_{x>0} \bigl(  xy +  \frac{CR^3y^4}{4} - \sqrt[3]{\frac{\beta x C^2}{(1 - \beta)}} y^3 R^2  + \bigl(\frac{\beta x C^{1/2}}{1 - \beta}\bigr)^{\frac{2}{3}} \frac{3Ry^2}{2}\bigr) \bigr).
\end{split}
\end{equation}
Let 
\begin{equation*}
f(x,y) = xy +  \frac{CR^3y^4}{4} - \sqrt[3]{\frac{\beta x C^2}{(1 - \beta)}} y^3 R^2  + \bigl(\frac{\beta x C^{1/2}}{1 - \beta}\bigr)^{\frac{2}{3}} \frac{3Ry^2}{2}.
\end{equation*}
Substituting $x=t^3$ we get,
\begin{equation*}
f(t,y) = t^3y +  \frac{CR^3y^4}{4} - \sqrt[3]{\frac{\beta C^2}{(1 - \beta)}}  y^3 R^2 t + \bigl(\frac{\beta C^{1/2}}{1 - \beta}\bigr)^{\frac{2}{3}} \frac{3Ry^2}{2} t^2.
\end{equation*}
Now,
\begin{equation}
\label{eqn:first_der_f}
\frac{\partial f(t,y)}{\partial t} = 3 t^2y + 3\bigl(\frac{\beta C^{1/2}}{1 - \beta}\bigr)^{\frac{2}{3}} Ry^2 t - \sqrt[3]{\frac{\beta C^2}{(1 - \beta)}}  y^3 R^2,
\end{equation}
and
\begin{equation}
\label{eqn:second_der_f}
\frac{\partial^2 f(t,y)}{\partial t^2} = 6ty + 3\bigl(\frac{\beta C^{1/2}}{1 - \beta}\bigr)^{\frac{2}{3}} Ry^2.
\end{equation}
Consider $f(t,y)$ at some fixed $y>0$. From \eqref{eqn:first_der_f} and \eqref{eqn:second_der_f}, the function $f(t,y)$ has two stationary points, one of which is a local minimum and the other a local maximum. Let us denote the local minimum by $t_{min}(y)$ and the local maximum by $t_{max}(y)$. We have $t_{min}(y) > 0$ and $t_{max}(y) < 0$. Also as $t \rightarrow -\infty$, $f(t,y) \rightarrow -\infty$ and as $t \rightarrow \infty$, $f(t,y) \rightarrow \infty$. Thus over $t > 0$, the function $f(t,y)$ has a unique global minimum. Hence, there exists a unique $x > 0$ (corresponding to the local minimum of $f(t,y)$, $t_{min}(y)$) which attains the infimum in equation \eqref{eqn:MGF_Gbarx_bdd2}. Let $x^*(y)$ denote the $x$ which attains the infimum in equation \eqref{eqn:MGF_Gbarx_bdd2}. The minimum $x^*(y)$ is given by
\begin{equation}
\label{eqn:MGF_Gbarx_bdd3}
x^*(y) = \frac{R^3y^3}{8} \Bigl[ - \Bigl(\frac{\beta C^{0.5}}{1-\beta}\Bigr)^{\frac{2}{3}} + \sqrt{\Bigl(\frac{\beta C^{0.5}}{1-\beta}\Bigr)^{\frac{4}{3}} + \frac{4}{3} \Bigl(\frac{\beta C^2}{1-\beta}\Bigr)^{\frac{1}{3}} } \Bigr]^3.
\end{equation}
Substituting \eqref{eqn:MGF_Gbarx_bdd3} in \eqref{eqn:MGF_Gbarx_bdd2} gives us, 
\begin{equation*}
H(y) = e^{-\gamma(C,\beta) R^3 y^4},
\end{equation*}
where $\gamma(C,\beta)$ is a constant dependent on $C$ and $\beta$. In Figure \ref{fig:CCDF_CUBIC_plot}, we illustrate $H(y)$ and $\mathbb{P}(\overline{G}_x \geq y)$ for $x = 0, 0.1, 1$ (with $C = 0.4$, $\beta = 0.3$ and RTT, $R = 1$).
\begin{figure}
\centering
\includegraphics[scale=0.23,trim={4cm 0.5cm 4cm 0.3cm},clip]{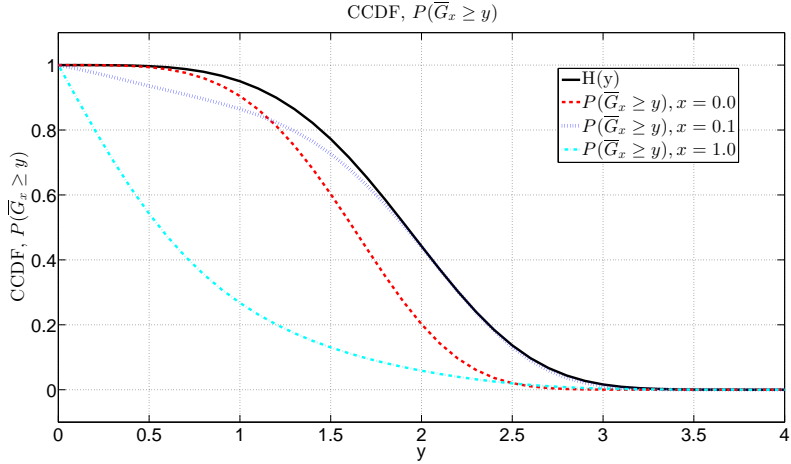}
\caption{TCP CUBIC: Upper bound for $\mathbb{P}(\overline{G}_x \geq y)$.}
\label{fig:CCDF_CUBIC_plot}
\end{figure}
For the version of TCP CUBIC we consider \cite{tcp_cubic_code}, we have $C=0.4$ and $\beta = 0.3$. For these values, $\gamma(C,\beta) = 0.0510 > 0$. (In fact, numerically evaluating $\gamma(C,\beta)$ we find that for $C=0.4$, for all $\beta \in (0,1)$, $\gamma(C,\beta) > 0$.) The function $H(y)$ is a complementary cumulative distribution function and is going down super-exponentially in $y$. Therefore, the moment generating function (MGF) corresponding to $H(y)$ is bounded in a neighborhood of $0$. Since $H(y)$ bounds $\mathbb{P}(\overline{G}_x \geq y)$ for all $x$,  we have
\begin{equation}
\sup_{x>0} \mathbb{E}[e^{t\overline{G}_x}] < \infty,
\end{equation}
for $t$ in some neighborhood of $0$.
\end{proof}

\begin{proposition}
\label{prop:CUBIC_Vbarx_invariance}
The Markov chain $\{ \overline{V}_n \} $ is Harris recurrent and has a unique invariant distribution.
\end{proposition}
\begin{proof}
We first prove that the Markov chain $\{ \overline{V}_n \} $ is Harris irreducible w.r.t. the Lebesgue measure on $\mathbb{R}^+$. To prove this, consider a point $x$ in the state space of $\{ \overline{V}_n \} $. Let $L(x,A)$ denote the probability of $\{ \overline{V}_n \} $ hitting set $A$ in a finite time starting with $\overline{V}_0 = x$. From equation \eqref{eqn:CUBIC_lim_Vn_def}, $\mathbb{P}(x,(x(1-\beta), \infty)) = 1$. The distribution of $\overline{G}_x$ is absolutely continuous with respect to the Lebesgue measure on $\mathbb{R}^+$. Therefore, for any set $A$ with non-zero Lebesgue measure, such that $A \subseteq ((1 - \beta) x, \infty)$,  $\mathbb{P}(x,A) > 0$. Hence, $L(x,A) > 0$. Also, $P(x, (x(1-\beta), x)) > 0$. Therefore, for any set $B \subseteq (0, (1 - \beta ) x)$ with non-zero Lebesgue measure, there exists $n$ such that $\mathbb{P}^n(x,B) > 0$. Therefore for any set, $C$ with non-zero Lebesgue measure, $L(x,C) > 0$. Thus the Markov chain $\{ \overline{V}_n \} $ is Harris irreducible w.r.t. the Lebesgue measure.

To show the positive recurrence of the Markov chain, we use a result from \cite[p. 116]{Kalashnikov1993}. In our setup, it is sufficient to prove the following results. There exists a $x^*$ such that
\begin{enumerate}
\item $\mathbb{E}[\overline{V}_{n+1} - \overline{V}_{n} | \overline{V}_{n+1} = x] \leq -\epsilon$, for all $x > x^*$ for some $\epsilon > 0$ .
\item $\mathbb{E}[\overline{V}_{n+1} | \overline{V}_{n} = x] < \infty$, for all $x \leq x^*$.
\end{enumerate}
From equation \eqref{eqn:CUBIC_lim_Vn_def},
\begin{equation}
\label{eqn:meandrift1}
\begin{split}
\mathbb{E}[\overline{V}_{n+1} - \overline{V}_{n} | \overline{V}_{n} = x]  &= \mathbb{E} \Bigl[(1 - \beta) C(\overline{G}_{x} R - K(x))^3  \Bigr] \\
& \leq (1 - \beta) C ( \mathbb{E}[\overline{G}_{x}^3] R^3 + \mathbb{E}[\overline{G}_{x}] R K(x)^2 - K(x)^3 ) \\
& \leq (1 - \beta) C ( \sup_{x>0} (\mathbb{E}[\overline{G}_{x}^3]) R^3 + \sup_{x>0} (\mathbb{E}[\overline{G}_{x}]) R K(x)^2 \\ 
& \hspace*{5cm} - K(x)^3 ). 
\end{split}
\end{equation}
Using Taylor's series expansion, we have
\begin{equation}
\label{eqn:fourthpower}
\frac{t^4 \overline{G}_x^4}{4!} \leq e^{t\overline{G}_x} + e^{-t\overline{G}_x}.
\end{equation}
From Lemma \ref{lemma:MGF_Gbarx_bdd}, $\sup_{x>0} \mathbb{E}[e^{t\overline{G}_x}]$ $< \infty$ for $t \in (-t_0, t_0)$ for some $t_0 > 0$. Therefore, using \eqref{eqn:fourthpower}, $\displaystyle{\sup_{x>0} (\mathbb{E}[\overline{G}_{x}^4])}$ is finite. We have $\overline{G}_{x}^3 \leq \overline{G}_{x}^4 + 1$ and $\overline{G}_{x} \leq \overline{G}_{x}^4 + 1$. Hence, the terms $\displaystyle{\sup_x (\mathbb{E}[\overline{G}_{x}^3])}$ and $\displaystyle{\sup_x (\mathbb{E}[\overline{G}_{x}])}$ are finite.  Therefore, the RHS of the last inequality in \eqref{eqn:meandrift1} has terms $x$ and $x^{\frac{2}{3}}$. Since the $x$ term dominates $x^{\frac{2}{3}}$ for large $x$ and has a negative coefficient, we can find an $x^*$ and an $\epsilon > 0$ such that for all $x > x^*$, 
\begin{equation}
\label{eqn:meandrift2}
\mathbb{E}[\overline{V}_{n+1} - \overline{V}_{n} | \overline{V}_{n+1} = x] \leq -\epsilon.
\end{equation}
Also for $x < x^*$, we have 
\begin{equation}
\label{eqn:meandrift3}
\begin{split}
\mathbb{E}[\overline{V}_{n+1} | \overline{V}_{n} = x]  &= \mathbb{E} \bigl[(1 - \beta) \bigl( C(\overline{G}_{x} R - K(x)^3 \bigr) \bigr] + x \\
& \leq (1 - \beta) C \bigl( \mathbb{E}[\overline{G}_{x}^3] R^3 + \mathbb{E}[\overline{G}_{x}] R K(x)^2 \bigr) + x^* \\
& \leq (1 - \beta) C \bigl( \sup_x (\mathbb{E}[\overline{G}_{x}^3]) R^3 + \sup_x (\mathbb{E}[\overline{G}_{x}]) R K(x^*)^2 \bigr) + x^* \\
& < \infty.
\end{split}
\end{equation}
The last inequality comes from Lemma \ref{lemma:MGF_Gbarx_bdd}. From equations \eqref{eqn:meandrift2} and \eqref{eqn:meandrift3}, we see that the Markov chain $\{ \overline{V}_n \} $ is Harris recurrent and has a unique invariant distribution.
\end{proof}

\subsection{Asymptotic Approximations for Throughput}
\label{subsec:asymptoticEW}
Let us denote the average time between losses, under stationarity, by $\mathbb{E}[G_{V(\infty)}^p]$. The average number of packets sent between two consecutive losses is $p^{-1}$. Therefore, by Palm calculus \cite{Asmussen}, the average window size, $E[W(p)]$, under stationarity, is  
\begin{equation*}
\mathbb{E}[W(p)] = \frac{1}{p\mathbb{E}[G_{V_{\infty}}^p]}.
\end{equation*}
From Proposition \ref{prop:CUBIC_finiteEW}, $\mathbb{E}[W(p)]$ is finite. 

Let $\overline{G}_{\overline{V}_{\infty}}$ denote the random variable with same distribution as the stationary distribution of $\overline{G}_{\overline{V}_{n}}$ (existence of which is proved by Proposition \ref{prop:CUBIC_Vbarx_invariance}). Then, from Proposition \ref{prop:Gbarx_CUBIC}, we expect $\mathbb{E}[p^{\frac{1}{4}} G_{V_{\infty}}^p]$ to be close to $\mathbb{E}[\overline{G}_{\overline{V}_{\infty}}]$ for small $p$ and hence
\begin{equation*}
\mathbb{E}[ W(p)] \approx \frac{p^{-\frac{3}{4}}}{\mathbb{E}[\overline{G}_{\overline{V}_{\infty}}]},
\end{equation*}
for small $p$. 

We can evaluate $\mathbb{E}[\overline{G}_{\overline{V}_{\infty}}]$ using Monte-Carlo simulations. In Figure \ref{fig:CUBIC_Gbar_time_avg}, we illustrate simulations for evaluating $\sum_{i=1}^{n} \frac{\overline{G}_{\overline{V}_{i}}}{n}$ with initial conditions $\overline{V}_{0} =$, $0.0$, $0.1$, $2.0$ for TCP CUBIC with parameters, $C = 0.4, \beta = 0.3$ (used in \cite{Ha2008}) and with RTT, $R = 1$ seconds. We see that in these cases, after $n > 250$, there is little change in $n^{-1} \sum_{i=1}^{n} \overline{G}_{\overline{V}_{i}}$. For the TCP CUBIC parameters as given above with RTT, $R = 1$, and using $n = 10000$, we get $\mathbb{E}[\overline{G}_{\overline{V}_{\infty}}] \approx  0.7690$. 
\begin{figure}
\centering
\includegraphics[scale=0.22]{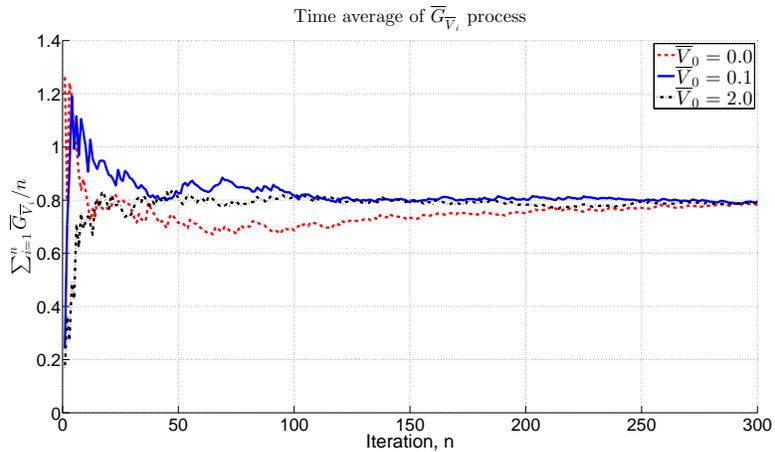}
\caption{Monte-Carlo simulations for $\mathbb{E}[\overline{G}_{\overline{V}_{\infty}}] \approx \sum_{i=1}^{n} \frac{\overline{G}_{\overline{V}_{i}}}{n}$}
\label{fig:CUBIC_Gbar_time_avg}
\end{figure}
Therefore for small $p$, we can approximate the average window size,  $\mathbb{E}[W(p)]$ for TCP CUBIC with $R = 1$ as $1.3004p^{-\frac{3}{4}}$.

From \eqref{eqn:time_between_losses_CUBIC}, the time average window size $\mathbb{E}[W(p)]$ for the fluid model is directly proportional to $R^\frac{3}{4}$. We evaluate $\mathbb{E}[W(p)]$ (approximately using Monte Carlo simulations) for TCP CUBIC with different RTT and find that $\mathbb{E}[W(p)]$ for our Markov chain approximation is also directly proportional to $R^\frac{3}{4}$.
Thus, we can approximate the average window size for TCP CUBIC flow with RTT $R$ as $\mathbb{E}[W(p)] = 1.3004 \Bigl(\frac{R}{p}\Bigr)^{\frac{3}{4}}$.
To account for the TCP Reno mode of operation which we ignored in our approximation, we make the following approximation,
\begin{equation}
\label{eqn:tcp_cubic_dumas}
\mathbb{E}[W(p)] = \max \bigl\{ 1.3004 \bigl(\frac{R}{p}\bigr)^{\frac{3}{4}}, \frac{1.31}{\sqrt{p}}\bigr\},
\end{equation}
where the second term on the RHS approximates the TCP Reno average window size as given in \cite{Mathis1997}. For the same parameters, the deterministic periodic loss model gives us $\mathbb{E}[W(p)] = \max \bigl\{ 1.0538 \bigl(\frac{R}{p}\bigr)^{\frac{3}{4}}, \frac{1.31}{\sqrt{p}}\bigr\}$. In the next section, we compare our model against the deterministic periodic loss model using ns2 simulations.

%% file: simulation_results.tex
\section{Simulation Results}
\label{sec:simulation_results}
In Tables \ref{tbl:cubic_dumas_EW} and \ref{tbl:cubic_dumas_G}, we compare the above approximation against the fluid approximate model in \cite{Ha2008}, our earlier Markov model in \cite{Sudheer2011} and ns2 simulations. Table \ref{tbl:cubic_dumas_EW} compares the average window size, whereas Table \ref{tbl:cubic_dumas_G} compares the goodput. The link speeds are set to $10$ Gbps so that the queuing is negligible. The packet sizes are  $1050$ bytes which is the default value in ns2. The maximum window size is set to $40000$. We see that, unlike TCP Reno, the average window size for TCP CUBIC depends on the RTT of the flow and increases with RTT. This behaviour makes TCP CUBIC fairer to flows with larger RTT as compared to TCP Reno and also leads to TCP CUBIC being more efficient over large-delay networks. 

For the deterministic loss model, we use  $\mathbb{E}[W(p)] = \max \{ 1.0538 \bigl(\frac{R}{p}\bigr)^{\frac{3}{4}}, \frac{1.31}{\sqrt{p}}\}$ so as to account for the Reno-mode of operation. If the RTT is small ($<0.1$ sec), TCP CUBIC operates more like Reno. In such cases, the deterministic loss model has accuracy similar to the Markovian models. However, when RTT is large ($>0.1$ sec) the Markovian models are better (sometimes much better especially for $R = 1$ sec) than the deterministic periodic loss model. The Markov model in \cite{Sudheer2011} explicitly considers the TCP-Reno mode behaviour, while here we just use a simple approximation to account for the TCP-Reno mode behaviour. However, in spite of this we see that the Markov model in \cite{Sudheer2011} performs only marginally better than the current Markovian approximation that we use in this paper. When compared against ns2 simulations, the Markov model in \cite{Sudheer2011} typically has errors $<4\%$, whereas for the current approximation given in this paper, the errors are $<5\%$ for most cases. 

\begin{table}
\caption{Average Window size via different approximations.}
\centering
\scalebox{0.85}{
\begin{tabular}{|c|c|c|c|c|c|}
\hline
per & RTT &$\mathbb{E}[W]$ & $\mathbb{E}[W]$ & $\mathbb{E}[W]$ & $\mathbb{E}[W]$ \\
\hline
$p$ & $R$ & Simulations  & Det. Fluid & Markov chain & Approx. Markov  \\
& (sec) & (ns2) & \cite{Ha2008} & \cite{Sudheer2011} & $\max \Bigl\{ 1.3004 \Bigl(\frac{R}{p}\Bigr)^{\frac{3}{4}}, \frac{1.31}{\sqrt{p}}\Bigr\}$ \\
\hline
$\num{1e-2}$ & $1$ & $39.97$ & $33.33$ & $37.44$ & $41.19$ \\
$\num{1e-2}$ & $0.2$ & $14.3$ & $13.10$ & $13.53$ & $13.10$ \\ 
$\num{1e-2}$ & $0.1$ & $12.62$ & $13.10$ & $12.50$ & $13.10$ \\ 
$\num{1e-2}$ & $0.02$ & $12.08$ & $13.10$ & $12.41$ & $13.10$ \\ 
$\num{1e-2}$ & $0.01$ & $11.53$ & $13.10$ & $12.41$ & $13.10$ \\ 
\hline
$\num{5e-3}$ & $1$ & $69.46$ & $56.05$ & $63.78$ & $69.27$ \\ 
$\num{5e-3}$ & $0.2$ & $21.82$ & $18.53$ & $21.02$ & $20.81$ \\ 
$\num{5e-3}$ & $0.1$ & $18.29$ & $18.53$ & $18.09$ & $18.53$ \\ 
$\num{5e-3}$ & $0.02$ & $17.21$ & $18.53$ & $17.73$ & $18.53$ \\ 
$\num{5e-3}$ & $0.01$ & $16.58$ & $18.53$ & $17.73$ & $18.53$ \\
\hline
$\num{1e-3}$ & $1$ & $229.96$ & $187.40$ & $218.32$ & $231.63$ \\ 
$\num{1e-3}$ & $0.2$ & $67.83$ & $56.05$ & $67.92$ & $69.58$ \\ 
$\num{1e-3}$ & $0.1$ & $44.68$ & $41.43$ & $44.55$ & $41.43$ \\ 
$\num{1e-3}$ & $0.02$ & $39.40$ & $41.43$ & $39.94$ & $41.43$ \\ 
$\num{1e-3}$ & $0.01$ & $38.71$ & $41.43$ & $39.94$ & $41.43$ \\ 
\hline
$\num{5e-4}$ & $1$ & $384.43$ & $315.17$ & $370.12$ & $388.56$ \\ 
$\num{5e-4}$ & $0.2$ & $113.05$ & $94.26$ & $114.52$ & $117.02$ \\ 
$\num{5e-4}$ & $0.1$ & $69.12$ & $58.58$ & $70.05$ & $69.24$ \\ 
$\num{5e-4}$ & $0.02$ & $55.89$ & $58.58$ & $56.66$ & $58.58$ \\ 
$\num{5e-4}$ & $0.01$ & $55.38$ & $58.58$ & $56.66$ & $58.58$ \\ 
\hline
$\num{8e-5}$ & $1$ & $1507.19$ & $1245.81$ & $1487.19$ & $1539.87$ \\ 
$\num{8e-5}$ & $0.2$ & $430.49$ & $372.58$ & $454.41$ & $462.57$ \\ 
$\num{8e-5}$ & $0.1$ & $260.91$ & $221.54$ & $271.15$ & $273.69$ \\ 
$\num{8e-5}$ & $0.02$ & $143.99$ & $146.46$ & $143.42$ & $146.46$ \\ 
$\num{8e-5}$ & $0.01$ & $140.83$ & $146.46$ & $142.71$ & $146.46$  \\ 
\hline
\end{tabular}
}
\label{tbl:cubic_dumas_EW}
\end{table}

\begin{table}
\caption{Goodput obtained via different approximations.}
\centering
\scalebox{0.915}{
\begin{tabular}{|c|c|c|c|c|c|}
\hline
per & RTT &$\lambda$ & $\lambda$ & $\lambda$ & $\lambda$ \\
\hline
$p$ & $R$ & Simulations  & Det. Fluid & Markov chain & Approx. Markov  \\
& (sec) & (ns2) & \cite{Ha2008} & \cite{Sudheer2011} & $\frac{E[W(p)]}{R}$ \\
\hline
$\num{1e-2}$ & $1$ & $39.55$ & $32.99$ & $37.06$ & $40.78$ \\
$\num{1e-2}$ & $0.2$ & $70.54$ & $64.85$ & $66.99$ & $64.85$ \\ 
$\num{1e-2}$ & $0.1$ & $124.5$ & $129.69$ & $123.7$ & $129.69$ \\ 
$\num{1e-2}$ & $0.02$ & $595.41$ & $648.45$ & $614.29$ & $648.45$ \\ 
$\num{1e-2}$ & $0.01$ & $1135.12$ & $1296.9$ & $1228.58$ & $1296.9$ \\ 
\hline
$\num{5e-3}$ & $1$ & $69.09$ & $55.77$ & $63.46$ & $68.93$ \\ 
$\num{5e-3}$ & $0.2$ & $108.43$ & $92.19$ & $104.56$ & $103.53$ \\ 
$\num{5e-3}$ & $0.1$ & $181.75$ & $184.37$ & $180.04$ & $184.37$ \\ 
$\num{5e-3}$ & $0.02$ & $854.30$ & $921.87$ & $881.97$ & $921.87$ \\ 
$\num{5e-3}$ & $0.01$ & $1645.59$ & $1843.74$ & $1763.94$ & $1843.74$ \\
\hline
$\num{1e-3}$ & $1$ & $226.72$ & $187.21$ & $218.10$ & $231.40$ \\ 
$\num{1e-3}$ & $0.2$ & $338.76$ & $279.95$ & $339.25$ & $347.46$ \\ 
$\num{1e-3}$ & $0.1$ & $446.28$ & $413.89$ & $445.09$ & $413.89$ \\ 
$\num{1e-3}$ & $0.02$ & $1966.64$ & $2069.43$ & $1994.96$ & $2069.43$ \\ 
$\num{1e-3}$ & $0.01$ & $3862.36 $ & $4138.87$ & $3989.91$ & $4138.86$ \\ 
\hline
$\num{5e-4}$ & $1$ & $384.16$ & $315.01$ & $369.94$ & $389.37$ \\ 
$\num{5e-4}$ & $0.2$ & $564.90$ & $471.05$ & $572.30$ & $584.81$ \\ 
$\num{5e-4}$ & $0.1$ & $690.72$ & $585.51$ & $700.16$ & $692.04$ \\ 
$\num{5e-4}$ & $0.02$ & $2791.22$ & $2927.54$ & $2831.66$ & $2927.54$ \\ 
$\num{5e-4}$ & $0.01$ & $5528.73$ & $5855.07$ & $5663.38$ & $5855.07$ \\ 
\hline
$\num{8e-5}$ & $1$ & $1506.79$ & $1245.71$ & $1487.07$ & $1539.75$ \\ 
$\num{8e-5}$ & $0.2$ & $2151.97$ & $1862.77$ & $2271.86$ & $2312.64$ \\ 
$\num{8e-5}$ & $0.1$ & $2608.56$ & $2215.23$ & $2711.30$ & $2736.66$ \\ 
$\num{8e-5}$ & $0.02$ & $7195.12$ & $7322.54$ & $7170.32$ & $7322.54$ \\ 
$\num{8e-5}$ & $0.01$ & $14067.18$ & $14645.07$ & $14270.00$ & $14645.07$  \\ 
\hline
\end{tabular}
}
\label{tbl:cubic_dumas_G}
\end{table}

Using the approach described in Section \ref{sec:markovmodel}, we can compute expressions for average window size for different TCP CUBIC parameters to study effect of these parameters on TCP performance. In Table \ref{tbl:cubic_dumas_beta_pt2}, we compare the results for $\beta = 0.2$ which is used by an older version of TCP CUBIC and is also widely used \cite{Yang2014}. For this parameter setting, we get 
\begin{equation}
\label{eqn:cubic_dumas_beta_pt2}
\mathbb{E}[W(p)] = \max \Bigl\{ 1.54 \Bigl(\frac{R}{p}\Bigr)^{\frac{3}{4}}, \frac{1.31}{\sqrt{p}}\Bigr\}.
\end{equation}
In this case, the median errors for the deterministic loss model in \cite{Ha2008}, for the Markov chain model in \cite{Sudheer2011} and the current approximation are $11\%$, $4.5\%$ and $6.2\%$ respectively.

\begin{table}
\caption{Comparison of different approximations for $\beta = 0.2$.}
\scalebox{0.8}{
\begin{tabular}{|c|c|c|c|c|c|}
\hline
per & RTT &$\mathbb{E}[W]$ & $\mathbb{E}[W]$ & $\mathbb{E}[W]$ & $\mathbb{E}[W]$ \\
\hline
$p$ & $R$ & Simulations  & Det. Fluid & Markov chain & Approx. Markov  \\
& & (ns2) & \cite{Ha2008} & \cite{Sudheer2011} & $\max \Bigl\{ 1.54 \Bigl(\frac{R}{p}\Bigr)^{\frac{3}{4}}, \frac{1.31}{\sqrt{p}}\Bigr\}$ \\
\hline
$\num{1e-2}$ & $1$ & $45.41$ & $37.13$ & $42.67$ & $48.69$ \\
$\num{1e-2}$ & $0.2$ & $15.52$ & $13.10$ & $14.53$ & $14.56$ \\ 
$\num{1e-2}$ & $0.1$ & $13.09$ & $13.10$ & $12.61$ & $13.10$ \\ 
$\num{1e-2}$ & $0.02$ & $12.09$ & $13.10$ & $12.5$ & $13.10$ \\ 
$\num{1e-2}$ & $0.01$ & $11.59$ & $13.10$ & $12.5$ & $13.10$ \\ 
\hline
$\num{5e-3}$ & $1$ & $78.43$ & $62.44$ & $72.69$ & $81.89$ \\ 
$\num{5e-3}$ & $0.2$ & $24.56$ & $18.67$ & $23.46$ & $24.49$ \\ 
$\num{5e-3}$ & $0.1$ & $19.06$ & $18.53$ & $18.31$ & $18.53$ \\ 
$\num{5e-3}$ & $0.02$ & $17.28$ & $18.53$ & $17.71$ & $18.53$ \\ 
$\num{5e-3}$ & $0.01$ & $16.69$ & $18.53$ & $17.71$ & $18.53$ \\
\hline
$\num{3e-3}$ & $1$ & $116.52$ & $91.59$ & $107.53$ & $120.12$ \\ 
$\num{3e-3}$ & $0.2$ & $35.13$ & $27.39$ & $34.22$ & $35.92$ \\ 
$\num{3e-3}$ & $0.1$ & $25.47$ & $23.92$ & $24.36$ & $23.92$ \\ 
$\num{3e-3}$ & $0.02$ & $22.46$ & $23.92$ & $22.88$ & $23.92$ \\ 
$\num{3e-3}$ & $0.01$ & $21.74$ & $23.92$ & $22.88$ & $23.92$ \\ 
\hline
\end{tabular}
}
\label{tbl:cubic_dumas_beta_pt2}
\end{table}

\subsection{Extension to Multiple TCP connections}
The expression \eqref{eqn:tcp_cubic_dumas} was obtained assuming the RTT to be constant, i.e., we assumed that the queuing was negligible. However when multiple TCP connections go through a link, the queuing may not be non-negligible. We can approximate the average window size for TCP CUBIC in this case by replacing $R$ with $\mathbb{E}[R]$ in \eqref{eqn:tcp_cubic_dumas}.

For illustration, we consider an example with $15$ flows sharing a single bottleneck link of speed $100$ Mbps. The packet error rates for five of these flows is $0.01$, it is $0.001$ for another group of five flows and for the rest it is set to $0.0001$. The propagation delays for the five connections in each group are set to $0.05$, $0.1$, $0.2$, $0.25$, and $0.5$ sec respectively. The average RTT of these flows is computed using the M/G/1 approximation from \cite{Poojary2016}, where the bottleneck queue is assumed to be an M/G/1 queue. In Figure  \ref{fig:15flows_CUBIC_throughput}, we plot a scatterplot comparing the throughput obtained using the M/G/1 approximation with ns2 simulations. The packet sizes were set to $1050$ bytes. The model approximations and simulations differ by $< 5\%$ for most cases and the maximum difference is $17\%$.

\begin{figure}
\centering
\includegraphics[scale=0.16]{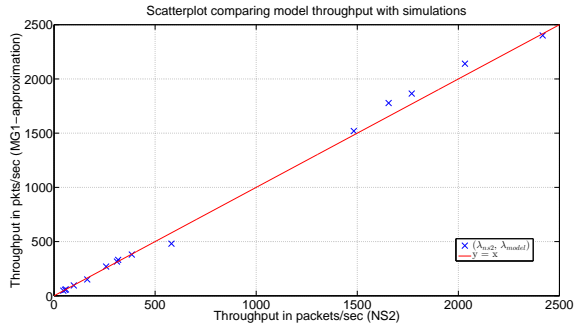}
\caption{Scatter plot for TCP CUBIC throughput ($15$ flows, bottleneck link speed: $100$Mbps)}
\label{fig:15flows_CUBIC_throughput}
\end{figure}

%% file: conclusion.tex
\section{Conclusion}
\label{sec:conclusion}
We have derived throughput expression for a single TCP CUBIC connection with fixed RTT under random losses. To this end, we first considered throughput expression developed for deterministic loss model for TCP CUBIC. We then considered the sequence of TCP window size processes indexed by $p$, the drop rate. We show that with appropriate scaling, this sequence converges to a limiting Markov chain. The scaling is obtained using insights from the deterministic loss model. The stationary distribution of the limiting Markov chain is then used to compute the desired throughput expression. We validate our model and assumptions by comparison with ns2 simulations. The ns2 simulations show a better match with our theoretical model as compared to the deterministic loss model.

%% file: appendix.tex
\begin{appendix}
\section{Appendix}
\label{app:appendixA}
\begin{proposition}
\label{prop:unif_conv}
Let $\{X_p(x), x \in \mathbb{R}^+\}$, be a process, for $0 < p < 1$, which converges to a limiting process $X(x)$ uniformly in the sense,
\begin{equation}
\label{eqn:unifconv_1}
\lim_{p \rightarrow 0} \sup_{x, y \leq M} \Bigl| \mathbb{P}(X_p(x) \leq y) - \mathbb{P}(X(x) \leq y) \Bigr| = 0,
\end{equation}
for any finite $M$, and for each $x$, let the limiting distribution, $\mathbb{P}(X(x) \leq y)$ be continuous. Then,
\begin{equation}
\lim_{p \rightarrow 0} \sup_{x \leq M} \Bigl| \mathbb{E}f(X_p(x)) -  \mathbb{E}f(X(x)) \Bigr| = 0,
\end{equation}
for any $f:\mathbb{R^+}\rightarrow\mathbb{R}$ continuous with compact support.
\end{proposition}
\begin{proof}
Consider a continuous function $f$ with compact support, $[0,K]$. Such a function is uniformly continuous. Therefore, given any $\epsilon$, there exists $m$ points $u_0 = 0 < u_1 < \cdots <  u_m = K$, such that 
\begin{equation}
\label{eqn:f_unif_cont}
\sup_{u_i< y < u_{i+1}} |f(y) - f(u_i)| < \epsilon,
\end{equation} 
for all $i=1,2,\cdots, m$.
We have 
\begin{equation} 
E[f(X_p(x)] = \int\limits_{0}^{K} f(u) \mathbb{P}(X_p(x) \in du). \\ 
\end{equation} 
From \eqref{eqn:f_unif_cont},
\begin{equation*}
\Bigl| E[f(X_p(x)] - \sum\limits_{i=1}^{m-1} f(u_i) \mathbb{P}(X_p(x) \in (u_i,u_{i+1}]) \Bigr| \leq \epsilon.
\end{equation*}
Similarly,
\begin{equation*} 
\begin{split} 
\Bigl| E[f(X(x)] - \sum\limits_{i=1}^{m-1} f(u_i) \mathbb{P}(X(x) \in (u_i,u_{i+1}])\Bigr| \leq \epsilon. 
\end{split}
\end{equation*}
Therefore, 
\begin{equation*} 
\begin{split} 
\Bigl|\mathbb{E}f(X_p(x) - \mathbb{E}f(X(x)) \Bigr| &\leq  \sum\limits_{i=1}^{m} f(u_i)  \Bigl| \mathbb{P}(X_p(x) \in (u_i,u_{i+1}]) \\
& \hspace*{12mm} - \mathbb{P}(X(x) \in  (u_i,u_{i+1}])\Bigr| + 2 \epsilon \\
&\leq  \sum\limits_{i=1}^{m}    \parallel f  \parallel_{\infty}  \Bigl|\mathbb{P}(X_p(x) \in (u_i,u_{i+1}]) \\ 
& \hspace*{12mm} - \mathbb{P}(X(x) \in (u_i,u_{i+1}])\Bigl| + 2 \epsilon,
\end{split}
\end{equation*} 
where $\parallel f \parallel_{\infty} = \sup \{f(x): x \in [0,K] \}$. Since $f$ is continuous over a compact support, it is bounded and hence $ \parallel f \parallel_{\infty} < \infty$.
Therefore
\begin{equation*} 
\begin{split} 
\lim_{p \rightarrow 0} & \sup_x \Bigl|\mathbb{E}f(X_p(x) - \mathbb{E}f(X(x)) \Bigr|  \\ 
& \leq  \lim_{p \rightarrow 0} \parallel f \parallel_{\infty} \sum\limits_{i=1}^{m} \sup_x \Bigl| \mathbb{P}(X_p(x) \in (u_i,u_{i+1}]) - \mathbb{P}(X(x) \in  (u_i,u_{i+1}]) \Bigr| + 2 \epsilon \\
& = 2 \epsilon.
\end{split}
\end{equation*}
The second relation follows from the hypothesis \eqref{eqn:unifconv_1}. Since $\epsilon$ is arbitrary we get the desired result.
\end{proof}

\end{appendix}

%% file: CUBIC_Markov_Approx_JNL.bbl
\begin{thebibliography}{10}
\providecommand{\url}[1]{#1}
\csname url@samestyle\endcsname
\providecommand{\newblock}{\relax}
\providecommand{\bibinfo}[2]{#2}
\providecommand{\BIBentrySTDinterwordspacing}{\spaceskip=0pt\relax}
\providecommand{\BIBentryALTinterwordstretchfactor}{4}
\providecommand{\BIBentryALTinterwordspacing}{\spaceskip=\fontdimen2\font plus
\BIBentryALTinterwordstretchfactor\fontdimen3\font minus
  \fontdimen4\font\relax}
\providecommand{\BIBforeignlanguage}[2]{{%
\expandafter\ifx\csname l@#1\endcsname\relax
\typeout{** WARNING: IEEEtran.bst: No hyphenation pattern has been}%
\typeout{** loaded for the language `#1'. Using the pattern for}%
\typeout{** the default language instead.}%
\else
\language=\csname l@#1\endcsname
\fi
#2}}
\providecommand{\BIBdecl}{\relax}
\BIBdecl

\bibitem{Huston2006}
G.~Huston, ``{Gigabit TCP},'' \emph{Internet Protocol Journal}, 2006.

\bibitem{rfc3649}
S.~Floyd, ``{HighSpeed TCP for Large Congestion Windows},'' RFC 3649
  (Experimental), Internet Engineering Task Force, December 2003.

\bibitem{Afanasyev2010}
A.~Afanasyev, N.~Tilley, P.~Reiher, and L.~Kleinrock, ``Host-to-host congestion
  control for tcp,'' \emph{Communications Surveys Tutorials, IEEE}, vol.~12,
  no.~3, pp. 304 --342, quarter 2010.

\bibitem{Yang2014}
P.~Yang, J.~Shao, W.~Luo, L.~Xu, J.~Deogun, and Y.~Lu, ``{TCP congestion
  avoidance algorithm identification},'' \emph{Networking, IEEE/ACM
  Transactions on}, vol.~22, no.~4, pp. 1311--1324, 2014.

\bibitem{Bonald1999}
T.~Bonald, ``{Comparison of TCP Reno and TCP Vegas: efficiency and fairness},''
  \emph{Performance Evaluation}, vol. 36-37, pp. 307 -- 332, 1999.

\bibitem{Liu2003}
Y.~Liu, F.~Lo~Presti, V.~Misra, D.~Towsley, and Y.~Gu, ``{Fluid Models and
  Solutions for Large-scale IP Networks},'' in \emph{Proceedings of the 2003
  ACM SIGMETRICS International Conference on Measurement and Modeling of
  Computer Systems}, ser. SIGMETRICS '03.\hskip 1em plus 0.5em minus
  0.4em\relax New York, NY, USA: ACM, 2003, pp. 91--101.

\bibitem{Vojnovic2000}
M.~Vojnovic, J.-Y. Le~Boudec, and C.~Boutremans, ``Global fairness of
  additive-increase and multiplicative-decrease with heterogeneous round-trip
  times,'' in \emph{INFOCOM 2000. Nineteenth Annual Joint Conference of the
  IEEE Computer and Communications Societies. Proceedings. IEEE}, vol.~3.\hskip
  1em plus 0.5em minus 0.4em\relax IEEE, 2000, pp. 1303--1312.

\bibitem{Kunniyur2003}
S.~Kunniyur and R.~Srikant, ``{End-to-end congestion control schemes: Utility
  functions, random losses and ECN marks},'' \emph{Networking, IEEE/ACM
  Transactions on}, vol.~11, no.~5, pp. 689--702, 2003.

\bibitem{Reddy2004}
V.~V. Reddy, V.~Sharma, and M.~B. Suma, ``{Providing QoS to TCP and Real Time
  Connections in the Internet},'' \emph{Queueing Syst. Theory Appl.}, vol.~46,
  no. 3/4, pp. 461--480, Mar. 2004.

\bibitem{Mathis1997}
M.~Mathis, J.~Semke, J.~Mahdavi, and T.~Ott, ``{The macroscopic behavior of the
  TCP congestion avoidance algorithm},'' \emph{SIGCOMM Comput. Commun. Rev.},
  vol.~27, pp. 67--82, July 1997.

\bibitem{Padhye2000}
J.~Padhye, V.~Firoiu, D.~F. Towsley, and J.~F. Kurose, ``{Modeling TCP Reno
  Performance: A Simple Model and Its Empirical Validation},'' \emph{IEEE/ACM
  Transactions on Networking}, vol.~8, pp. 133--145, 2000.

\bibitem{Xue2014}
L.~Xue, S.~Kumar, C.~Cui, and S.-J. Park, ``{A study of fairness among
  heterogeneous TCP variants over 10Gbps high-speed optical networks},''
  \emph{Optical Switching and Networking}, vol.~13, pp. 124--134, 2014.

\bibitem{Jain2011}
S.~Jain and G.~Raina, ``{An experimental evaluation of CUBIC TCP in a small
  buffer regime},'' in \emph{NCC, 2011}.\hskip 1em plus 0.5em minus 0.4em\relax
  IEEE, 2011, pp. 1--5.

\bibitem{Weigle2006}
M.~C. Weigle, P.~Sharma, and J.~Freeman~{IV}, ``Performance of competing
  high-speed {TCP} flows,'' in \emph{Proceedings of NETWORKING}, Coimbra,
  Portugal, may 2006, pp. 476--487.

\bibitem{Sudheer2011}
S.~Poojary and V.~Sharma, ``{Analytical Model for Congestion Control and
  Throughput with TCP CUBIC Connections.}'' in \emph{GLOBECOM}.\hskip 1em plus
  0.5em minus 0.4em\relax IEEE, 2011.

\bibitem{Bao2010}
W.~Bao, V.~W.~S. Wong, and V.~C.~M. Leung, ``{A Model for Steady State
  Throughput of TCP CUBIC},'' in \emph{GLOBECOM}.\hskip 1em plus 0.5em minus
  0.4em\relax IEEE, 2010.

\bibitem{Belhareth2013}
S.~Belhareth, L.~Sassatelli, D.~Collange, D.~Lopez-Pacheco, and
  G.~Urvoy-Keller, ``Understanding tcp cubic performance in the cloud: A
  mean-field approach,'' in \emph{Cloud Networking (CloudNet), 2013 IEEE 2nd
  International Conference on}, Nov 2013, pp. 190--194.

\bibitem{Ha2008}
S.~Ha, I.~Rhee, and L.~Xu, ``{CUBIC: a new TCP-friendly high-speed TCP
  variant},'' \emph{SIGOPS Oper. Syst. Rev.}, vol.~42, pp. 64--74, July 2008.

\bibitem{Sudheer2013Allerton}
S.~Poojary and V.~Sharma, ``{Approximate theoretical models for TCP connections
  using different high speed congestion control algorithms in a multihop
  network},'' in \emph{Communication, Control, and Computing (Allerton), 2013
  51st Annual Allerton Conference on}, Oct 2013, pp. 559--566.

\bibitem{Dumas2002}
V.~Dumas, F.~Guillemin, and P.~Robert, ``{A Markovian analysis of
  additive-increase multiplicative-decrease algorithms},'' \emph{Advances in
  Applied Probability}, vol.~34, no.~1, pp. 85--111, 2002.

\bibitem{Tan2006Infocom}
K.~Tan, J.~Song, Q.~Zhang, and M.~Sridharan, ``{A Compound {TCP} Approach for
  High-Speed and Long Distance Networks},'' in \emph{IEEE Infocom}, 2006.

\bibitem{Tweedie1983}
R.~Tweedie, ``{The existence of moments for stationary Markov chains},''
  \emph{Journal of Applied Probability}, pp. 191--196, 1983.

\bibitem{tcp_cubic_code}
S.~Ha and S.~Hemminger, ``{TCP CUBIC source code, ver. 2.3},''
  \url{http://lxr.free-electrons.com/source/net/ipv4/tcp_cubic.c?v=4.2}, {Last
  accessed: 08 Dec 2015}.

\bibitem{Kalashnikov1993}
V.~V. Kalashnikov, \emph{Mathematical methods in queuing theory}.\hskip 1em
  plus 0.5em minus 0.4em\relax Springer Science \& Business Media, 1993, vol.
  271.

\bibitem{Asmussen}
S.~Asmussen, \emph{Applied probability and queues}, 2nd~ed., ser. Applications
  of Mathematics (New York).\hskip 1em plus 0.5em minus 0.4em\relax New York:
  Springer-Verlag, 2003, vol.~51, stochastic Modelling and Applied Probability.

\bibitem{Poojary2016}
S.~Poojary and V.~Sharma, ``{Analysis of multiple flows using different high
  speed TCP protocols on a general network},'' \emph{Performance Evaluation},
  vol. 104, pp. 42 -- 62, 2016.

\end{thebibliography}
